\documentclass[aps,pra,twocolumn,amsmath,amssymb,showpacs,floatfix,superscriptaddress]{revtex4-1}

\usepackage{times}

\usepackage[utf8x]{inputenc}
\usepackage[english]{babel}
\usepackage[T1]{fontenc}
\usepackage{lmodern}
\usepackage{amssymb}
\usepackage{amsmath}
\usepackage{amsthm}
\usepackage[pdftex]{graphicx}
\usepackage{epstopdf}
\usepackage{braket}
\usepackage[bottom]{footmisc}

\usepackage{dcolumn}
\usepackage{bm}
\usepackage{color}
\usepackage{relsize,dsfont,mathrsfs,empheq,verbatim,upgreek}
\usepackage{etoolbox}
\usepackage[caption = false]{subfig}

\newtheorem{mylem}{Lemma}
\newtheorem{mythe}{Theorem}

\newcommand{\Tr}{\mathrm{Tr}}
\newcommand{\Lt}{\mathcal{L}}

\newcommand{\hilbert}{\mathscr{H}}

\newcommand{\ketbra}[2]{\ket{#1}\!\bra{#2}}

\newcommand{\Abs}[1]{\left|#1\right|}
\newcommand{\im}{\mathrm{i}}

\begin{document}


\title{Initial Correlations in Open Quantum Systems: Constructing Linear Dynamical Maps and Master Equations}


\author{Alessandra Colla}

\affiliation{Institute of Physics, University of Freiburg, 
Hermann-Herder-Stra{\ss}e 3, D-79104 Freiburg, Germany}

\author{Niklas Neubrand}

\affiliation{Institute of Physics, University of Freiburg, 
Hermann-Herder-Stra{\ss}e 3, D-79104 Freiburg, Germany}

\author{Heinz-Peter Breuer}

\affiliation{Institute of Physics, University of Freiburg, 
Hermann-Herder-Stra{\ss}e 3, D-79104 Freiburg, Germany}

\affiliation{EUCOR Centre for Quantum Science and Quantum Computing,
University of Freiburg, Hermann-Herder-Stra{\ss}e 3, D-79104 Freiburg, Germany}

\begin{abstract}
We investigate the dynamics of open quantum systems which are initially correlated with their environment. The strategy of our approach is to analyze how given, fixed initial correlations modify the evolution of the open system with respect to the corresponding uncorrelated dynamical behavior with the same fixed initial environmental state,
described by a completely positive dynamical map. We show that, for any predetermined initial correlations, one can introduce a linear dynamical map on the space of operators of the open system which acts like the proper dynamical map on the set of physical states and represents its unique linear extension. Furthermore, we demonstrate that this construction leads to a linear, time-local quantum master equation with generalized Lindblad structure involving time-dependent, possibly negative transition rates. Thus, the general non-Markovian dynamics of an open quantum system can be described by means of a time-local master equation even in the case of arbitrary, fixed initial system-environment correlations. We present some illustrative examples and explain the relation of our approach to several other approaches proposed in the literature.
\end{abstract}

\date{\today}

\maketitle

\section{Introduction} 
The theory of open quantum systems has since its conception gained more and more traction as a powerful tool for studying quantum systems \cite{Breuer2007}. Its formulation for the case of systems weakly coupled to Markovian baths has been extended to arbitrary coupling and non-Markovian behavior with numerous techniques, which brought to the formulation of popular exact master equations such as the Nakajima-Zwanzig equation \cite{Nakajima1958,Zwanzig1960} and the time-convolutionless (TCL) master equation \cite{Shibata1977,Chaturvedi1979}. An often made assumption in many approaches is that of factorizing initial condition between the system and the environment. 
Such a requirement has been highly criticized, either on the grounds of it being unphysical \cite{Munro1996,Kampen2004,Kampen2005}, or just too restrictive, and the question of initial correlation has become of importance \cite{Pechukas1994,Alicki1995,Pechukas1995}. For example, the treatment of the dynamics of open systems including correlations in the initial state leads to a general method for the local detection of correlations between the system and an inaccessible environment \cite{Laine2010b,Gessner2011a} which has been realized experimentally both in trapped ion and in photonic systems \cite{Gessner2014a, Cialdi2014,Gessner2019}. Several efforts have thus been made to extend the theory to allow for correlated initial states, both in terms of examining conditions for which the resulting dynamical map is completely positive \cite{Shabani2009,Brodutch2013,Liu2014,Buscemi2014,Shabani2016,Vacchini2016,Schmid2019} and of finding alternative methods for the dynamical description of the reduced system \cite{Modi2012,Ringbauer2015,PazSilva2019,Alipour2020,Trevisan2021}. The approaches are numerous and varied, and the wide array of papers on the subject paints the picture of a complicated subject matter.  Considering initial correlations can mean very different things, depending on what one assumes of the initial system-environment state: should the environment be in a fixed, specific initial state, or should it depend on the correlations? Are the system and the environment initially in a separable, classically correlated, entangled state? Different choices lead to different reduced dynamics, and may induce limitations such as loss of complete positivity or a restriction of the set of initial system states that can be studied. Trying to avoid these choices to allow for any initial total states might instead result in other losses, like that of a unique dynamical map \cite{PazSilva2019}. Ultimately, the decision should be made based on which questions one is trying to answer or which advantage one wants to gain, carefully taking into account the consequent drawbacks.

A possible limitation of initial correlations between the system and the environment is that they may impose non-linearity of the equations describing the evolution of the reduced system, the dynamical map and the master equation, depending on the initial conditions chosen -- for example, such non-linearity may appear in the shape of additional inhomogeneous terms \cite{Stelmachovic2001,Zhang2015}.  However, dynamical maps in the case of entangled initial state have been extended to linear maps on matrices by performing a basis transformation \cite{Jordan2004,Jordan2005}, showing that such a linearized map is, as a consequence, no longer completely positive. Some groups, primarily Dominy, Shabani and Lidar (DSL) \cite{Dominy2015,Dominy2016}, have then more rigorously approached the subject of initial correlations by developing general frameworks for which the linearity of the dynamical map can be preserved or obtained. The linearity property, in fact, comes with several advantages, such as the operator-sum representation due to Choi \cite{Choi1975} of linear and Hermiticity preserving superoperators and the subsequent criterion for such a superoperator to be completely positive, the renowned Kraus representation \cite{Kraus1983}. Moreover, a dynamical map which is linear -- even if not completely positive (CP) -- automatically leads, through its inverse, to an exact generator of the evolution which can be put into a time-local generalized Lindblad form with time-dependent, possibly negative rates \cite{Breuer2012,Hall2014}. Master equations of this form have a wide range of applications, e.g. in non-Markovian stochastic unravelings \cite{Breuer1999b,Breuer2004,Piilo2008a,Smirne2020} and recently proposed formulations of quantum thermodynamics \cite{Colla2022}, and make a strong case for why finding a linear dynamical map should be preferred.  But there is a trade-off: in order to maintain linearity, these frameworks have to sacrifice the universality of states studied, both in the sense of what class of total initial states are considered, and of how many initial reduced states can be described via the same evolution.  While the DSL framework has been proven to be the most general for the study of initial correlations while maintaining linearity \cite{Sargolzahi2020},  in its generality one might find it difficult to get a feeling for the operational procedure involved in practice, and for how big are the sacrifices made for the safeguard of linearity.

In this work, we focus on a view of initial correlations which is not intended to keep track of the \textit{kind} of correlations initially present, but rather focuses on the difference between uncorrelated and correlated initial states, by looking at how the evolution depends on the initial correlation operator. While this has been seen in other approaches to give rise to an affine dynamical map \cite{Stelmachovic2001,Zhang2015}, this context can be easily embedded into a formalism which is in spirit analogous to the DSL approach. In fact, as we shall see through an explicit prescription, it is always possible to construct a linear dynamical map for the reduced system, as a consequence of the fact that the dynamical map itself acts linearly on the set of density matrices (it does so even on the set of trace 1 matrices) and can thus be extended to a linear map on the set of all operators. The resulting evolution is then comprised of the usual ``uncorrelated'',  completely positive and trace preserving (CPT) dynamical map plus an extra part which depends on the initial correlations. From this we can construct the associated master equation, which is once again linear, and exists any time the inverse of the original uncorrelated dynamical map exists. Thanks to the linearity of the master equation, this can be put in generalized Lindblad form, with initial correlations contributions appearing only in the dissipator. In our view, our results show how initial correlations do not represent such an added conceptual challenge with respect to the uncorrelated initial state, contrary to what is often held, the only drawbacks being the possible loss of complete positivity and the restriction on the initial reduced states that can be studied with the same dynamical map. 

The structure of the paper is the following: in Sec.~\ref{sec:affine} we review how the assumption of a fixed initial correlation operator leads to an affine dynamical map, and formally define the domain of validity of such map. In Sec.~\ref{sec:linear} we explicitly construct the unique linear extension for the dynamical map and establish a formal criterion for complete positivity. In Sec.~\ref{sec:me} we deduce the linear master equation accounting for initial correlations, and study the structural change with respect to the uncorrelated case.  Sec.~\ref{sec:example} contains an application of the proposed approach to the Jaynes-Cummings model. We make concluding remarks in Sec.~\ref{sec:concl}. 

\section{Initial correlations and affine dynamical maps}\label{sec:affine}
For an uncorrelated initial state of the total system $\rho_{SE}(0)= \rho_S(0)\otimes \rho_E$ the dynamical map $\Phi_t$ describing the evolution of the reduced system is determined by two factors only:
\begin{enumerate}
\item the unitary evolution for the total system $U_t$
\item the initial environmental state $ \rho_E$
\end{enumerate} as one can see directly from the definition of the density matrix of the reduced system
\begin{equation}
\rho_S(t) \overset{\chi=0}{=} \Phi_t[\rho_S(0)]= \Tr_E \{ U_t \rho_S(0)\otimes \rho_E U_t^{\dagger} \}\;,
\end{equation}
where the superscript $\chi=0$ denotes the absence of initial correlations. Throughout the paper, we assume the Hilbert space of the reduced system to be finite dimensional.
The map $\Phi_t$ is linear, trace preserving and always completely positive (CPT), thus admitting a Kraus representation
\begin{equation}\label{CPKraus}
\Phi_t[\rho_S(0)]= \sum_i \Omega_i (t) \rho_S(0) \Omega_i^{\dagger}(t) \;,
\end{equation}
through the set of time dependent operators $\Omega_i$ satisfying $\sum_i \Omega_i^{\dagger}\Omega_i = \mathbb{I}$ at all times. For a correlated initial state $\rho_{SE}(0) \neq \rho_S(0)\otimes \rho_E$ most of these results fail. However,  one can analogously study the general case and compare it to its uncorrelated counterpart by dividing any initial total system state into its corresponding uncorrelated state and a correlation operator $\chi$:
\begin{equation}\label{initialstate}
\rho_{SE}(0)= \rho_S(0)\otimes \rho_E + \chi \;,
\end{equation}
where $ \rho_S(0)= \Tr_E\{\rho_{SE}(0)\}$,  $ \rho_E = \Tr_S \{\rho_{SE}(0)\}$ are the respective reduced states and $\chi$ has the property of being Hermitian and of yielding the null operator for both partial traces,  $\Tr_E\{\chi\}=0$, $\Tr_S\{\chi\}=0$. Now, the dynamical map depends additionally on the initial correlation operator through an extra term:
\begin{equation}\label{affinedynamicalmap}
\Phi^{\chi}_t[\rho_S(0)]= \Phi_t[\rho_S(0)] + I_t^{\chi}\;,
\end{equation}
with $I_t^{\chi} = \Tr_E \{ U_t \chi U_t^{\dagger} \}$.  In the special case where $\chi$ commutes with the Hamiltonian generating the unitary evolution $U_t$, namely when the correlation operator is left invariant by the evolution, the dynamical map above is identical to the uncorrelated $\Phi_t$.

Our strategy in order to deal with correlated initial states is to determine the reduced dynamics through a dynamical map determined by
\begin{enumerate}
\item the unitary evolution for the total system $U_t$
\item the initial environmental state $ \rho_E$
\item the correlation operator $\chi$ ,
\end{enumerate}
effectively regarding the initial correlation operator as a parameter to be taken independently of the reduced system state, analogously to what is done to $\rho_E$ in the uncorrelated case. Note that these assumptions are also at the core of \cite{Stelmachovic2001}, where the parameters describing the environment and the correlations at time zero should be determined (fixed). Indeed one can always do so; however, the set of initial reduced states whose evolution can be adequately described is then limited to the set of ``physical'' states $\mathcal{P}_E^{\chi}(\hilbert_S)$, dependent on $\rho_E$ and $\chi$, for which the total operator \eqref{initialstate} is still a proper state of the total Hilbert space. A general and explicit characterization of this set is not straightforward, as it heavily depends on the interplay between the chosen environmental state and the correlation operator. Nonetheless, a formal definition for $\mathcal{P}^{\chi}_E(\hilbert_S)$ can be given by introducing an assignment map \cite{Pechukas1994}. This map is designed to map a reduced state into a unique system-environment state. For our case, the assignment map depends on the choice of initial environmental state and correlations:
\begin{eqnarray}
\alpha_E^{\chi}& : &\mathcal{B}(\hilbert_S) \longrightarrow \mathcal{B}(\hilbert_{SE}) \\
& &X \longmapsto X\otimes \rho_E + \chi \; .
\end{eqnarray}
Using the notation $\mathcal{S}(\hilbert_{SE})$ to denote the convex set of states of the full Hilbert space, the physical domain $\mathcal{P}^{\chi}_E(\hilbert_S)$ can be identified as the preimage of $\mathcal{S}(\hilbert_{SE})$ under the assignment map $\alpha_E^{\chi}$, i.e.
\begin{equation}
\mathcal{P}^{\chi}_E(\hilbert_S) = {\left(\alpha_E^{\chi}\right) }^{-1} \left[ \mathcal{S}(\hilbert_{SE}) \right] \;.
\end{equation}
In more explicit terms, the physical domain corresponds to the elements of $\mathcal{S}(\hilbert_{S})$ for which the corresponding total operator is positive, i.e. any state $\rho_S \in \mathcal{S}(\hilbert_{S})$ satisfying the condition
\begin{equation}
\rho_S \otimes \rho_E + \chi \geq 0 \; .
\end{equation}
Naturally, for $\chi=0$ one recovers $\mathcal{P}^{\chi=0}_E(\hilbert_S)= \mathcal{S}(\hilbert_{S})$. 

Once $\chi$ is set, then the dynamical map \eqref{affinedynamicalmap} becomes an affine map on the space of bounded operators $\mathcal{B}(\hilbert_S)$, with a CPT linear component and a traceless offset term. It is thus still trace-preserving, but not necessarily completely positive.  While it is true that formally $\Phi^{\chi}_t$ is well defined on all $\mathcal{B}(\hilbert_S)$, it only assumes physical meaning when acting on elements of $\mathcal{P}_E^{\chi}(\hilbert_S)$. As we will see in the next section, this allows us to replace \eqref{affinedynamicalmap} with an equivalent linear map defined on all bounded operators and that acts as the proper dynamical map on $\mathcal{P}_E^{\chi}(\hilbert_S)$. 

\section{Linear dynamical map for initial correlations}\label{sec:linear}
Let us define the set of all bounded operators of the reduced Hilbert space that have trace one, $\mathcal{A}^1(\hilbert_S) := \{ X \in \mathcal{B}(\hilbert_S)\; | \;\Tr \{X\}=1\}$. This is an affine subspace of $\mathcal{B}(\hilbert_S)$, since all affine combinations (combinations with coefficients $\lambda_i \in \mathbb{C}$ such that $\sum_i \lambda_i =1$) of trace one operators $X_i \in \mathcal{A}^1(\hilbert_S)$ are still of trace one:
\begin{equation}
\Tr \Big\{ \sum_i \lambda_i X_i \Big\} = \sum_i \lambda_i =1 \; .
\end{equation}
We recall that an affine map can also be defined when acting on an affine space as a map that is linear under affine combinations; i.e., a map $f$ acting on an affine space $\mathcal{A}$ is affine if and only if, for any set of elements $A_i \in \mathcal{A}$ and any set of coefficients $\lambda_i$ such that $\sum_i \lambda_i = 1$, it follows that
\begin{equation}
f \Big( \sum_i \lambda_i A_i \Big) = \sum_i \lambda_i f(A_i) \; .
\end{equation}
It is known that such an affine map can be uniquely extended to a linear map on the smallest linear space containing $\mathcal{A}$, denoted by $\mathrm{Span}(\mathcal{A})$. It follows from this reasoning that the dynamical map \eqref{affinedynamicalmap}, which acts linearly over all trace one operators $\mathcal{A}^1(\hilbert_S)$, can be uniquely extended to a linear map on $\mathrm{Span}(\mathcal{A}^1(\hilbert_S)) = \mathcal{B}(\hilbert_S)$, the full space of bounded operators for the reduced system. These are straightforward mathematical results, whose proofs we nonetheless report for completeness in Appendix \ref{app:proofs}, both in abstract terms and for the specific case studied.
The important point is that the dynamical map for initial correlations \eqref{affinedynamicalmap} can be substituted with a unique equivalent linear map on all bounded operators that still describes the proper evolution of all relevant states. One can easily check that the following map
\begin{equation}\label{lineardynamicalmap}
\Psi^{\chi}_t[X]= \Phi_t[X] + I_t^{\chi} \Tr\{X\}
\end{equation}
has the wanted properties of being linear and extending \eqref{affinedynamicalmap} to all $X \in \mathcal{B}(\hilbert_S)$, and must therefore be its unique linear extension.

Let us examine its features.
Like in the affine case, the tracelessness of $I_t^{\chi}$ guarantees that the map is trace preserving. In general, with respect to the previous uncorrelated version, it instead loses the property of complete positivity and possibly even positivity. Still, the map evolves any element of the physical domain to a proper state of the reduced system, $\Psi^{\chi}_t[\mathcal{P}_E^{\chi}(\hilbert_S)] \subset \mathcal{S}(\hilbert_S)$, and it is written as a sum of a completely positive, correlation independent part and a term depending on initial correlations. Since this second term is linear and Hermiticity preserving, it admits a pseudo-Kraus representation of the following form \cite{Choi1975}:
\begin{equation}\label{pKrausChi}
 I_t^{\chi} \Tr\{\rho_S(0)\} = \sum_i f_i(t) F_i (t) \rho_S(0) F_i^{\dagger}(t) \; ,
\end{equation}
with the extra condition $\sum_i f_i(t)F_i^{\dagger}(t)  F_i (t) =0$ and where the coefficients $f_i(t)$ can be negative.  From the spectral decomposition of $I_t^{\chi}$,
\begin{equation}
 I_t^{\chi}  = \sum_j a_j(t) \ket{\varphi_j(t)}  \bra{\varphi_j(t)}\; ,
\end{equation}
one can recognize the operators and coefficients in the pseudo-Kraus representation \eqref{pKrausChi} to be
\begin{equation}
F_i(t) = \ket{\varphi_j(t)}  \bra{\varphi_{j'}(t)} \;  , \; \; f_i(t) = a_j(t) \; ,
\end{equation}
with $i$ a double index $\{j,j'\}$. To merge the two operator-sums \eqref{CPKraus} and \eqref{pKrausChi}, one can fix a basis of operators $E_k$ for which
\begin{equation}
F_i(t) = \sum_k \gamma_{ik}(t) E_k \; , \; \; \Omega_i(t) = \sum_k \omega_{ik}(t) E_k \; ,
\end{equation}
so that the full dynamical map can be written as
\begin{equation}
\Psi^{\chi}_t[\rho_S(0)]= \sum_{k, k'} \epsilon_{kk'} (t) E_k \rho_S(0) E^{\dagger}_{k'}  \; ,
\end{equation}
with
\begin{equation}
\epsilon_{kk'} (t) = \sum_i \Big( \omega_{ik} (t) \omega^*_{ik'}(t) + f_i(t) \gamma_{ik}(t) \gamma^*_{ik'} (t) \Big) \;.
\end{equation}
If the matrix $\epsilon$ is positive semi-definite, then the full map is CP.  This suggests that the condition for complete positivity is highly case dependent in this framework, as it depends on the interplay between $\rho_E$ and $\chi$, as well as on time through the evolution $U_t$.

\subsection*{Remarks: structure of initial correlations and different dynamical maps}\label{subsec:qubits}
It is important at this point to make some remarks for the sake of clarity. First, we mention again that within our approach one cannot study initial correlations in the sense of separable, classically correlated or entangled states; these names classify the \textit{kind} of correlations between system and environment, and by fixing $\rho_E$ and $\chi$ we describe possibly very different kinds of initial correlations depending on the value of $\rho_S(0)$.  While for certain values the total system may be in a pure state, for others it can be mixed, for certain it can have zero-discord,  and so on. 
As said before, we are not here interested in studying the effects of which kind of initial correlations have on the evolution of the reduced system, but are rather concerned with the effect of \textit{some} non-zero correlations with respect to the uncorrelated evolution.

We further remark that dynamical maps found through other methods and under different conditions -- say, by assuming a specific \emph{kind} of initial correlation -- will in general, whether linear or not, be different from the one proposed here, even if both maps are suitable for describing the evolution of the very same specific initial state. The reason is that making different assumptions on the total initial state imposes a different structure for the dynamical maps: they will have a different physical domain of applicability, and describe different situations at the level of the total initial state. Still, both dynamical maps will yield the same result at all times for any state for which the associated total state enters both assumptions.

We believe it useful to clarify this concept with a very simple example.  Take two qubits, each in basis $\{\ket{0},\ket{1}\}$, coupled through a swap gate described by the Hamiltonian:
\begin{equation}
H_{\mathrm{swap}}= {1\over2}(\mathbb{I} + \sigma_x \otimes \sigma_x +  \sigma_y \otimes \sigma_y +  \sigma_z \otimes \sigma_z ) \; ,
\end{equation}
and let us consider the first qubit as the reduced system of interest, and the second as the environment. In \cite{RodriguezRosario2008}, one can find an explicit prescription for the construction of the (completely positive) dynamical map that describes the evolution of a reduced system initially in a classically correlated (zero discord) state with its environment:
\begin{equation}
\rho_{SE} (0) = \sum_i p_i \Pi_i \otimes \rho_i \; ,
\end{equation}
where $\{p_i\}$ are probabilities, $\{\Pi_i\}$ are orthonormal projections and $\{\rho_i\}$ are independent states of the environment. The dynamical map one builds from \cite{RodriguezRosario2008} depends on the choice of the sets $\{\Pi_i\}$ and $\{\rho_i\}$. For the two qubits, let us choose the projections onto the states $\ket{0}$ and $\ket{1}$:
\begin{equation} \label{zdinitialstate}
\rho_{SE} (0) = p \ket{0}\bra{0} \otimes \rho_0 + (1-p) \ket{1}\bra{1} \otimes \rho_1 \; ,
\end{equation}
and the environmental states $\rho_0 = (\mathbb{I} -\sigma_x/2)/2$ and $\rho_1 = (\mathbb{I} +\sigma_x/2)/2$. The dynamical map for the ground state population and the coherences reads:
\begin{align}\label{zdmap1}
\rho_{00} (t) =&{1\over 2} \sin^2 (t) +\cos^2 (t) \rho_{00}(0) \; ,\\ \nonumber
\rho_{01} (t) =&{1\over 4} (\sin^2 (t) -i \sin (t) \cos (t))  \\ \label{zdmap2}
 &- {1\over 2}\sin^2 (t) \rho_{00}(0)  + {\sqrt{3}\over2}\cos^2 (t) \rho_{01}(0) \; .
\end{align}
Notice the following: the map does not depend on $p$, as this is precisely the varying parameter that determines $\rho_S(0)$, such that the domain $\mathcal{P}^{\text{zd}}(\hilbert_S)$ of the above dynamical map is given by states of the form $ p \ket{0}\bra{0}+ (1-p) \ket{1}\bra{1}$ with $0\leq p\leq 1$. With this construction there is no fixed environmental state, as this depends directly on the choice of $p$. What is fixed here is the structure -- zero discord, with this specific set of projections and environmental states -- of the total initial system, namely a particular kind of initial correlation.

Let us now view this situation in our formalism.  From a specific state of the form \eqref{zdinitialstate},  i.e. with a fixed $p$, we can extract the environmental state $\rho_E$ and the correlations $\chi$, all of which depend on $p$, and regard them as the fixed initial conditions of the total system. We can this way construct the dynamical map as previously described, which reads
\begin{align}\nonumber
\rho_{00} (t) =&{1\over 2} \sin^2(t) +\cos^2 (t) \rho_{00}(0)\\ \label{qubitsmap1}
& - {1\over 2} \sin (t) \cos (t) (2 p -1) \Im\rho_{01}(0) \; , \\ \nonumber
\rho_{01} (t) =& {1\over 4} (\sin^2 (t) -i \sin (t) \cos (t))- {1\over 2}\sin^2 (t) p  \\ \nonumber
&+ {i\over 2} \sin (t) \cos (t) (2 p -1)(p-\rho_{00}(0)) \\   \label{qubitsmap2}
&+ \cos^2 (t) \rho_{01}(0) \; .
\end{align}
This dynamical map now depends on $p$ and has a different structure than \eqref{zdmap1}-\eqref{zdmap2}, but also does not act upon the same states. The domain is given by all states $\rho_S(0)$ for which $\rho_S(0)\otimes \rho_E + \chi$ is still a state; in Fig.~\ref{fig:domains_ZD} we report examples of the domain $\mathcal{P}_E^{\chi}(\hilbert_S)$ for different values of $p$. For these initial states, $\rho_{SE}(0)$ is in general not zero-discord: even though we started with considering the zero-discord state \eqref{zdinitialstate}, our real assumptions are the fixed $\rho_E$ and $\chi$; relaxing $\rho_S(0)$ gives then different kinds of initial correlations. The one state that belongs to both domains $\mathcal{P}^{\text{zd}}(\hilbert_S)$ and $\mathcal{P}_E^{\chi}(\hilbert_S)$ and is associated to the same total state
\begin{equation}
\rho_S(0)\otimes\rho_E + \chi = p \ket{0}\bra{0} \otimes \rho_0 + (1-p) \ket{1}\bra{1} \otimes \rho_1
\end{equation}
is the one given by $\rho_{00} = p$ and $\rho_{01} =0$, which is also, as expected, the only case for which \eqref{zdmap1}-\eqref{zdmap2} and \eqref{qubitsmap1}-\eqref{qubitsmap2} coincide.
\begin{figure}[tp]
\includegraphics[width=1\columnwidth]{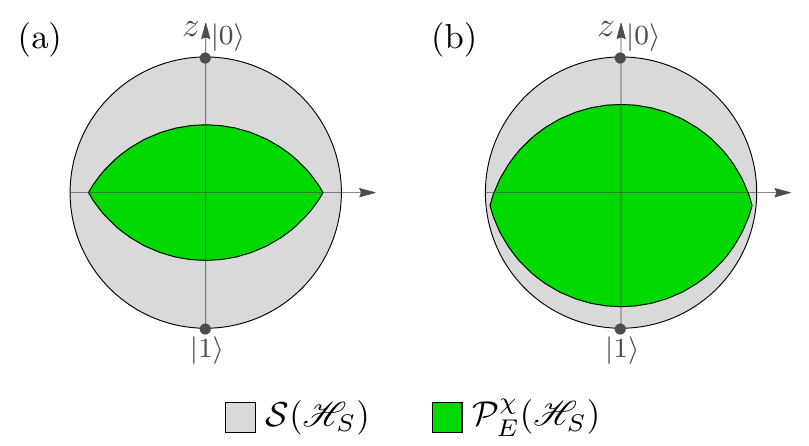}
\caption{Cross-section of the Bloch sphere through the $z$-axis showing the physical domain for the dynamical map \eqref{qubitsmap1}-\eqref{qubitsmap2}, which is determined by the values of $\rho_{00}$ and the modulus of $\rho_{01}$ only (thus axially symmetric around $z$), at values of (a) $p=1/2$ and (b) $p=7/8$, with respect to the set of states $\mathcal{S}(\hilbert_S)$.}
\label{fig:domains_ZD}
\end{figure}

As a last comment, we remark that the ``zero-discord'' dynamical map \eqref{zdmap1}-\eqref{zdmap2} does not reduce to the identity for the limit $t \rightarrow 0$. Of course it still is \emph{compatible} with the identity for the physical initial states that should be considered (in particular, it is compatible for $\rho_{01}(0)=0$), but the map itself is not the identity at $t=0$,  as opposed to our map \eqref{qubitsmap1}-\eqref{qubitsmap2}.  This can be considered a drawback for the purpose of constructing a master equation by means of a perturbative expansion. 

\section{Linear time-local master equation for initial correlations}\label{sec:me}
A main advantage of having a linear dynamical map is that the generator of the evolution can be put in generalized Lindblad form with time-dependent coefficients \cite{Breuer2012,Hall2014}. In the uncorrelated case, for example, whenever the inverse of the dynamical map exists, the generator can be written as
\begin{equation}
\mathcal{L}_t[X] := \dot{\Phi}_t\circ \Phi_t^{-1}[X] \; ,
\end{equation}
such that the uncorrelated exact master equation
\begin{equation}
\dot{\rho}_S(t) \overset{\chi=0}{=} \mathcal{L}_t[\rho_S(t)]
\end{equation}
follows. Since $\Phi_t$ is linear, then $\mathcal{L}_t$ is linear as well, and is therefore an element of the space of superoperators $\mathfrak{htp}(\hilbert_S)$ of Hermiticity and trace preserving (trace destroying, to be more precise) superoperators, i.e. which satisfy the conditions
\begin{equation}
 \Lt [X^{\dagger}] = \Lt [X]^{\dagger}, \quad \Tr \big\{\Lt [X]\big\} = 0 \quad  \forall X \in \mathcal{B}(\hilbert_S).
\end{equation}
As such, it can be written as a sum of a dissipator and a commutator with an effective Hamiltonian, i.e.
\begin{equation} \label{canonical}
 \Lt_t [X]= -i [K_S(t),X] + {\mathcal{D}}_t [X],
\end{equation}
where
\begin{equation} \label{Diss-part}
 {\mathcal{D}}_t [X] = \sum_{k}\lambda_{k}(t)\Big[L_{k}(t)
 X L_{k}^{\dag}(t) - \frac{1}{2}\big\{L_{k}^{\dag}(t)L_{k}(t),X \big\}\Big] \; ,
\end{equation}
for some rates $\lambda_k(t)\in \mathbb{R}$, operators $\{L_k(t)\}$ and Hermitian $K_S(t)$.

We can now extend this procedure for any correlated initial state using the linear dynamical map \eqref{lineardynamicalmap}. Regarding the existence of the associated generator, we see that, since the inverse dynamical map reads
\begin{equation}\label{inversemap}
{\left(\Psi^{\chi}_t\right)}^{-1}[X]= \Phi_t^{-1}[X] - \Phi_t^{-1}[I_t^{\chi}] \Tr\{X\}\; ,
\end{equation}
it follows that $\Psi^{\chi}_t$ is invertible if and only if $\Phi_t$ is invertible.  Therefore, the presence of initial correlations $\chi$ has no influence on the conditions for which the generator
\begin{equation}
\mathcal{L}^{\chi}_t[X] := \dot{\Psi}^{\chi}_t\circ {\left(\Psi^{\chi}_t\right)}^{-1}[X]
\end{equation}
exists, with respect to the uncorrelated case.
In any case, the generator explicitly reads:
\begin{equation}\label{corrgenerator}
\mathcal{L}^{\chi}_t [ X ] = \mathcal{L}_t[X] + \mathcal{J}^{\chi}_t \Tr \{ X\} \; ,
\end{equation}
with
\begin{equation}
\mathcal{J}^{\chi}_t  := \dot{I}_t^{\chi} -  \mathcal{L}_t[I_t^{\chi}]  \; .
\end{equation}
Once again for the master equation, the generalization to fixed initial correlations is given by the uncorrelated part plus an extra linear term that depends on $\chi$.  An equivalent strategy would have been, in fact, to derive an affine master equation -- which would have been identical to the one derived in \cite{Stelmachovic2001} -- from the affine dynamical map \eqref{affinedynamicalmap}, and linearize this directly through the reasoning we presented in the previous section.
We briefly remark that $\mathcal{J}^{\chi}_t$ also corresponds to the inhomogeneity one obtains for the affine TCL master equation through projection operator technique with projection $P[\rho_{SE}(t)] = \rho_S(t) \otimes \rho_E$ (see Appendix \ref{app:TCLinhomog} for a proof of this); one may therefore directly apply our proposed description to any of these results in order to obtain a linear master equation.

We shall like to see how the additional correlation term modifies equation \eqref{canonical}. We first notice that, equivalently to the possible loss of positivity of the dynamical map, the generator \eqref{corrgenerator} might no longer preserve positivity. The structure, instead, remains the same: since the superoperator $ \mathcal{J}^{\chi}_t \Tr \{ \cdot\} $ is an element of $\mathfrak{htp}(\hilbert_S)$,  it also admits a decomposition as a commutator with some Hamiltonian and a dissipator:
\begin{equation} \label{canonicalchi}
\mathcal{J}^{\chi}_t \Tr \{ X\}= -i [k^{\chi}_S(t),X] + d^{\chi}_t [X] \; . 
\end{equation}
The splitting into the two terms is made unique by choosing traceless Lindblad operators in the dissipator. An equivalent requirement is also proposed in \cite{Sorce2022}, along with a derivation of the  Hamiltonian and the dissipator starting from a pseudo-Kraus representation of the generator. For our superoperator of interest, analogously to \eqref{pKrausChi}, the pseudo-Kraus representation is given by
\begin{equation}\label{pKrausChi-me}
 \mathcal{J}^{\chi}_t \Tr \{X\} = \sum_i g_i(t) G_i (t) X G_i^{\dagger}(t) \; ,
\end{equation}
with $\sum_i g_i(t)G_i^{\dagger}(t)  G_i (t) =0$ and
\begin{equation}
G_i(t) = \ket{\eta_j(t)}  \bra{\eta_{j'}(t)} \;  , \; \; g_i(t) = b_j(t) \; ,
\end{equation}
with $i=\{j,j'\}$ a double index, and where $ \ket{\eta_j(t)}  $ are the eigenvectors of $ \mathcal{J}^{\chi}_t$ and $b_j(t)$ are its eigenvalues.
Following \cite{Sorce2022} for the expression of the Hamiltonian, we find that it happens to vanish:
\begin{eqnarray}\nonumber
k^{\chi}_S(t) &=& {1\over 2i \mathrm{d}_S} \sum_i g_i(t) \left[ \Tr\{G_i(t)\} G_i^{\dagger}(t) -  \Tr\{G_i^{\dagger}(t)\} G_i(t) \right] \\ \label{K_S^chi=0}
&=& 0\; ,
\end{eqnarray}
where $\mathrm{d}_S=\mathrm{dim}(\hilbert_S)$, since the operators $G_i(t)$ are either traceless or Hermitian. The dissipator instead reads 
\begin{equation} \label{dissipatorchi}
 d^{\chi}_t [X] = \sum_i g_i(t)\left[ J_i (t) X J_i^{\dagger}(t) - {1\over 2} \left\{ J_i^{\dagger}(t)  J_i (t), X \right\} \right] \; ,
\end{equation}
with traceless Lindblad operators
\begin{equation}
J_i (t) = G_i(t) - {  \Tr\{G_i(t)\}  \over \mathrm{d}_S} \mathbb{I} \; .
\end{equation}
The full generator is therefore given by 
\begin{equation} \label{canonicalcorr}
\mathcal{L}^{\chi}_t [ X ] = -i [K_S(t),X] + \mathcal{D}^{\chi}_t [X] \; ,
\end{equation}
with $ \mathcal{D}^{\chi}_t =  \mathcal{D}_t +  d^{\chi}_t $, and with the canonical Hamiltonian left unaltered by the presence of initial correlations.  Ultimately, the generalized master equation reads:
\begin{widetext}
\begin{equation} \label{mecorr}
\begin{aligned}
\dot{\rho}_S(t) =& -i [K_S(t),{\rho_S}(t)]  \\ 
& \quad + \sum_i\Bigg[\lambda_{i}(t)\left[L_{i}(t) {\rho_S}(t) L_{i}^{\dag}(t) - \frac{1}{2}\left\{L_{i}^{\dag}(t)L_{i}(t),{\rho_S}(t) \right\}\right] + g_i(t)\left[ J_i (t) {\rho_S}(t) J_i^{\dagger}(t) - {1\over 2} \left\{ J_i^{\dagger}(t)  J_i (t), {\rho_S}(t) \right\} \right] \Bigg]\; .
\end{aligned}
\end{equation}
\end{widetext}
Depending on the case, it might be easy or straightforward to merge the two dissipators appearing above into a unique expression. If the set of Lindblad operators $\{J_i\}$ is a subset of the uncorrelated set $\{L_i\}$, the contribution of initial correlations appears solely as a renormalization of the rates. This is the case of our example in Sec.~\ref{sec:example}.  

With the above we have shown that the presence of initial correlations does not change the structure of the time-local master equation. Since this procedure can be carried out for any initial correlation $\chi$,  this shows that the formal treatment of the evolution of the reduced system in terms of a linear and time local master equation as it is known for uncorrelated initial states can be extended to describe any initially correlated state of the total system, given that the initial correlation operator is known.

Let us finally remark that a Lindblad-type master equation including contributions from evolved system-environment correlations was derived in \cite{Alipour2020}; this equation is non-linear in the reduced density matrix, and thus cannot be put in the simple form we employ here -- at least, as it looks, not without compromising the aim of it to keep track of correlations as they evolve in time. We instead care for total-system quantities only at initial times: in the subsequent evolution, we are only concerned with system degrees of freedom, as in the original open quantum system approach. 

\section{Example}\label{sec:example}

In the following section the system of interest is given by a two-level system with Hamiltonian
\begin{align}
    H_S = \omega_0 \sigma_+\sigma_-\;,
\end{align}
where $\omega_0$ is the transition frequency and $\sigma_\pm$ the raising and lowering operators between the ground state $\ket{g}$ and the excited state $\ket{e}$.
This qubit is coupled to a monochromatic radiation field with Hamiltonian
\begin{align}
    H_E = \omega b^\dagger b\;,
\end{align}
frequency $\omega$, creation and annihilation operators $b^\dagger$ and $b$, via the interaction Hamiltonian
\begin{align}
    H_I =  g \left( \sigma_+\otimes b + \sigma_-\otimes b^\dagger \right)
\end{align}
with real coupling strength $g$.
The total Hamiltonian $H_S+H_E+H_I$ is known as the Jaynes-Cummings model and is exactly solvable for all initial states $\rho_{SE}(0)$ including all possible correlations~\cite{Jaynes1963, grynberg2010}.
For clarity reasons, however, we will not consider the most general case, but take one specific initial state of the total system as a reference,
\begin{align}\label{example2:modelstate}
    \rho_{SE} = p_0 \rho_0 \otimes \ketbra{0}{0} + (1-p_0)\rho_1 \otimes \ketbra{1}{1}\;,
\end{align}
with $0<p_0<1$. Here the two-level system states are taken as $\rho_0 = (\mathbb{I} + a \sigma_z)/2$ and $\rho_1 = (\mathbb{I} - a \sigma_z)/2$, mixed system states corresponding to opposite points on the $z$-axis of the Bloch sphere, with $0 \leq a \leq 1$. This is also a zero-discord state, as the environment states are taken to be the projectors onto the ground and the first excited state.
The corresponding reduced states of system and environment read
\begin{align}\label{rhos_example}
    \rho_S &= \frac{1}{2}\left(\mathbb{I} + a(2p_0-1)\sigma_z \right) \\
    \rho_E &= p_0 \ketbra{0}{0} + (1-p_0) \ketbra{1}{1}
\end{align}
and give rise to the correlation operator
\begin{align}\label{chi_JC}
    \chi &= p_0(1-p_0)a \sigma_z \otimes \left( \ketbra{0}{0} - \ketbra{1}{1} \right)\;.
\end{align}

As described in Sec.~\ref{sec:affine} we now consider total initial states of form $\rho_{SE}(0) = \rho_S(0)\otimes\rho_E+\chi$ with fixed $\rho_E$ and $\chi$, and construct the physical domain $\mathcal{P}_E^\chi(\hilbert_S)$ of $\rho_S(0)$ for which $\rho_{SE}(0)$ is a proper state of the total system.
Interestingly, $\rho_{SE}(0)$ can be easily recast in the form
\begin{align}
    \rho_{SE}(0) = p_0 \tau_0 \otimes \ketbra{0}{0} + (1-p_0) \tau_1 \otimes \ketbra{1}{1}\;,
\end{align}
with system operators $\tau_i = (\mathbb{I}+[\vec{v}-\vec{c}_i]\vec{\sigma})/2$, Bloch vector $\vec{v}$ of $\rho_S(0)$, $\vec{c}_0 = -2a(1-p_0) \vec{e}_z$ and $\vec{c}_1 = 2 a p_0 \vec{e}_z$.
It can be shown that $\rho_{SE}(0)$ is a state of the total system if and only if both $\tau_i$ are states of the reduced system.
This implies that $\vec{v}$ must satisfy the equations
\begin{align}
&|| \vec{v}-\vec{c}_0||^2 \leq 1\\ 
&|| \vec{v}-\vec{c}_1||^2 \leq 1\; ,
\end{align}
such that the physical domain can be interpreted in the Bloch sphere as the intersection of the unit spheres around $\vec{c}_0$ and $\vec{c}_1$, as shown in Fig.~\ref{fig:JC_domains}.
The center of the physical domain $\vec{r}=a(2p_0-1) \vec{e}_z$ is exactly the Bloch vector of the reference state $\rho_S$ and the volume depends on the parameter $a$. For the limiting case $a=1$ the physical domain reduces to the single point $\rho_S$, while for $a=0$ the reference state $\rho_{SE}$ becomes a product state, for which the physical domain is the entire Bloch sphere.

\begin{figure}
    \centering
    \includegraphics[width=\columnwidth]{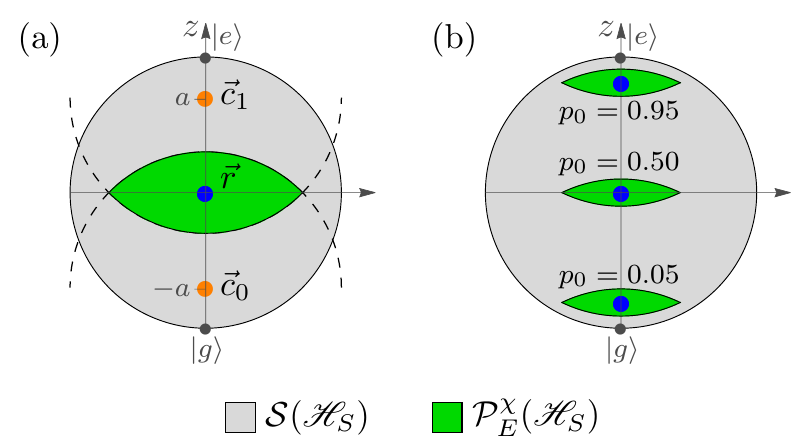}
    \caption{
    Physical domain of the dynamical map in a 2-dimensional representation of the Bloch sphere.
    (a) The physical domain $\mathcal{P}_E^\chi(\hilbert_S)$ is the intersection of two unit spheres (dashed) with centers $\vec{c}_0$ and $\vec{c}_1$. The Bloch vector $\vec{r}$ of the reduced reference state $\rho_{S}$ is the center of the physical domain. Here we choose parameters $p_0=0.5$ and $a=0.7$.
    (b) For fixed $a=0.9$ we show how the choice of $p_0$ determines the position of $\vec{r}$ and hence $\mathcal{P}_E^\chi(\hilbert_S)$.
    }
    \label{fig:JC_domains}
\end{figure}

To find the affine dynamical map we make use of the expression for the uncorrelated dynamical map found in Ref.~\cite{Smirne2010b} and calculate the inhomogeneity exploiting the expression for the time evolution operator $U(t)$ of the total system.
Since these expressions are given in the interaction picture with respect to $H_S+H_E$, we transform the results back to the Schrödinger picture and receive
\begin{align}
\label{example2:affinedynamicalmapgg}
    \rho_{gg} (t) &= \rho_{gg}(0)[\alpha(t)+\beta(t)-1] + 1 - \beta(t) - f(t) \\ \label{example2:affinedynamicalmapeg}
    \rho_{eg}(t) &= \rho_{eg}(0)e^{-\im\omega_0 t}\gamma(t)
\end{align}
with time-dependent coefficients
\begin{align}
    \alpha(t) &= 1 - (1-p_0) \Abs{d_1(t)}^2 \\
    \beta(t) &= 1 - p_0 \Abs{d_1(t)}^2 - (1-p_0) \Abs{d_2(t)}^2 \\
    \gamma(t) &= c_1(t) \left[ p_0 + (1-p_0)c_2(t) \right] \\
    f(t) &= ap_0(1-p_0) \Abs{d_2(t)}^2\;,
\end{align}
where the functions
\begin{align}
    c_n(t) &= e^{\im \Delta t / 2} \left[\cos\left(\frac{\Omega_n}{2}t\right) - \im \frac{\Delta}{\Omega_n} \sin\left(\frac{\Omega_n}{2}t\right) \right] \\
    |d_n(t)|^2 &= n\left(\frac{2g}{\Omega_n}\right)^2 \sin^2\left(\frac{\Omega_n}{2}t\right)\;.
\end{align}
oscillate with detuning frequency $\Delta=\omega_0-\omega$ and excitation-dependent Rabi frequencies $\Omega_n = \sqrt{\Delta^2+4g^2 n}$.
While the coefficients $\alpha(t)$, $\beta(t)$ and $\gamma(t)$ describe the uncorrelated part of the evolution, the effect of the initial correlations is encoded in the coefficient $f(t)$ coming from the inhomogeneity $I_t^\chi = f(t)\sigma_z$,  and leads to periodic additional excitations of the system. This is shown in Fig.~\ref{fig:dynamics_comparison} for a specific choice of model parameters. For the specific form of correlations \eqref{chi_JC} we have taken, the coherences are not influenced by the presence of the initial correlation. It is important to note, however, that a different choice of $\chi$ may in general lead to changes also in the off-diagonal elements.

\begin{figure}
    \centering
    \includegraphics[height=5cm]{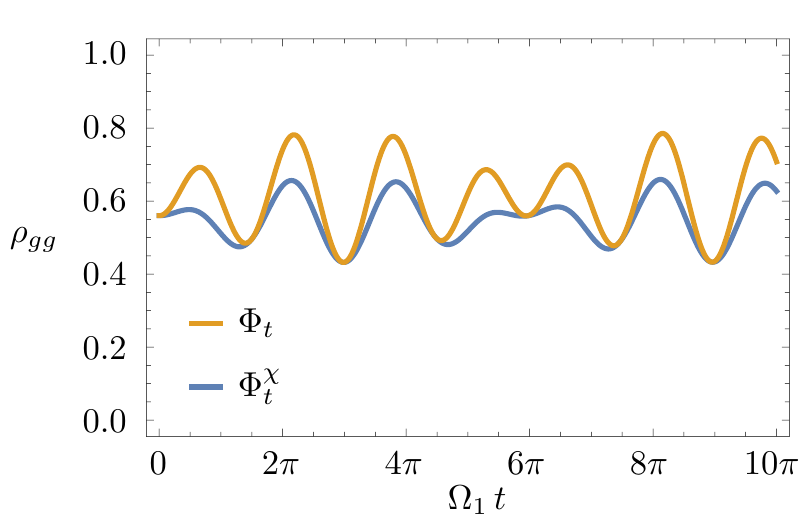}
    \caption{Comparison of the ground state probability $\rho_{gg}$ at time $t$ after evolution with dynamical maps $\Phi_t^\chi$ and $\Phi_t$ with and without the initial correlations $\chi$, respectively.
    For the correlated evolution $\rho_{gg}$ is reduced by the function $f(t)\geq0$ with respect to the uncorrelated evolution.
    We take $\rho_S(0)=\rho_S$, with $\rho_S$ as in Eq.~\eqref{rhos_example}, and parameters $a=0.6$, $p_0=0.4$, $\Delta=0.1\, \omega_0$ and $g=0.1\, \omega_0$.
    }
    \label{fig:dynamics_comparison}
\end{figure}

From the affine dynamical map we can also calculate the linearized dynamical map and finally the time-local master equation generator \eqref{canonicalcorr}
as described in Secs.~\ref{sec:linear} and~\ref{sec:me}.
Then the effective system Hamiltonian is given by the original Hamiltonian $H_S$ plus a time-dependent Lamb shift,
\begin{align}
    K_S(t) = \left[\omega_0 - \Im\left(\frac{\dot{\gamma}(t)}{\gamma(t)}\right) \right] \sigma_+\sigma_-\,,
\end{align}
and appears independent of the initial correlations as required by Eq.~\eqref{K_S^chi=0}.
The correlated dissipator $\mathcal{D}_t^\chi$ is the sum of the uncorrelated dissipator $\mathcal{D}_t$ from Ref.~\cite{Smirne2010b} and the dissipator $d_t^\chi$ induced by the correlations.
Interestingly,  the latter can be written using $\sigma_\pm$ as Lindblad operators,  which also appear in $\mathcal{D}_t$, such that we directly get a unified representation
\begin{align}
    \mathcal{D}_t^\chi[X] = \sum_{k = +,-,z} \lambda_k(t) \left[\sigma_k X \sigma_k^\dagger - \frac{1}{2} \{\sigma_k^\dagger\sigma_k, X\} \right]
\end{align}
with Lindblad rates
\begin{align}
    \lambda_+ &= \frac{(\alpha-f-1)\dot{\beta} - (\beta+f)\dot{\alpha} + (\alpha+\beta-1)\dot{f}}{\alpha+\beta-1} \\
    \lambda_- &= \frac{(\beta+f-1)\dot{\alpha} - (\alpha-f)\dot{\beta} - (\alpha+\beta-1)\dot{f}}{\alpha+\beta-1} \\
    \lambda_z &= \frac{1}{4} \left[ \frac{\dot{\alpha}+\dot{\beta}}{\alpha+\beta-1} - 2 \Re\left(\frac{\dot{\gamma}}{\gamma}\right) \right]\;,
\end{align}
where the time arguments have been omitted to increase readability.  Only $\lambda_+$ and $\lambda_-$ are modified by the initial correlations, and for the case $f(t)=0$ one recovers the master equation for the system under the assumption of factorizing initial conditions.

\section{Conclusions and applications}\label{sec:concl}
In this work we have explicitly shown how to uniquely linearize the affine dynamical map and time-local master equation that arise for a reduced system showing some fixed correlations with its environment at initial times. It turns out that the simple ``fix'' of multiplying the inhomogeneity by $\Tr\{\rho_S\}$ constitutes the unique linear extension for both the dynamical map and the master equation. This allowed us to generalize important tools for the case of initial correlations, while keeping track of the separate terms pertaining to the associated uncorrelated initial state and to the initial correlation operator.  In particular, we see that the correlation contribution in the master equation modifies only the dissipator and leaves invariant the Hamiltonian obtained in the uncorrelated case. The main drawback of this method, where the initial correlation operator must be known, is that the resulting dynamical map and master equation are only valid for a restricted set of initial reduced state $\mathcal{P}^{\chi}_E$ which defines the ``physical domain'' of these maps. The dynamical map is, in addition, in general not completely positive, although this does not affect the results on the dynamics and the master equation. The main advantage, however, is the gain of a linear, modified time-local master equation that contains initial correlations, and that exists any time the uncorrelated master equation does. With the help of such a master equation, popular techniques and applications that were previously only applicable under the assumption of uncorrelated initial states, such as stochastic unraveling and quantum thermodynamics, can be straightforwardly extended to account for initial correlations. 

Once one adopts the point of view described in this paper, it becomes evident that initial correlations really have a structural impact on the dynamics only when they are taken into account as an entity whose details and characterization are regarded as an important feature on its own. The mere fact that the system and environment are correlated, regardless of \emph{how} they are correlated, does not imply a formal change to the structure of the evolution of the reduced system with respect to the uncorrelated case, and most importantly does not challenge the applicability of techniques and theories which are built upon the idea of a linear evolution. 

\begin{acknowledgments}
We would like to thank Matteo Gallone and Bassano Vacchini for useful comments and discussion. This project has received funding from the European Union’s Framework Programme for Research and Innovation Horizon 2020 (2014-2020) under the Marie Skłodowska-Curie Grant Agreement No. 847471. We acknowledge support by the Open Access Publication Fund of the University of Freiburg.
\end{acknowledgments}

\begin{appendix}

\section{Proof of a unique linear extension}\label{app:proofs}
The statement of an affine map admitting a unique linear extension on the span of the affine subspace on which it acts might seem an obvious mathematical fact. We believe that it is nonetheless useful to see an explicit proof, both for the particular case at hand and in more abstract terms.
\subsection{Linear extension from trace one operators}
We start by analyzing the precise case studied in Sec.~\ref{sec:linear}, namely a linear extension from the set of trace one operators to the space of bounded operators. Let $\Phi : \mathcal{B}(\hilbert_S) \longrightarrow \mathcal{B}(\hilbert_S)$ be an affine map. It can therefore be decomposed as 
\begin{equation}
\Phi [X] = \Lt[X] + B \; ,
\end{equation}
where $\Lt$ is a linear map and $B \in \mathcal{B}(\hilbert_S)$ is a constant shift. The map 
\begin{equation}
\Psi [X] = \Lt[X] + B\Tr\{X\} 
\end{equation}
is a linear extension of $\Phi$ from the set of trace one operators $\mathcal{A}^1(\hilbert_S)$ to the full $\mathcal{B}(\hilbert_S)$. 
\begin{mythe}\label{theo:affine1}
Let $\Lambda : \mathcal{B}(\hilbert_S) \longrightarrow \mathcal{B}(\hilbert_S)$ be a linear map satisfying 
\begin{equation}
\Lambda[X]= \Phi[X] \; \; \forall X \in \mathcal{A}^1(\hilbert_S) \; .
\end{equation}
Then 
\begin{equation}\label{uniqueext}
\Lambda[X] = \Psi [X] \;\;\; \forall X \in \mathcal{B}(\hilbert_S)\; ,
\end{equation} 
i.e. $\Psi$ is the unique linear extension of $\Phi$.
\end{mythe}
\begin{proof}
Assume first that $\Tr \{ X\} \neq 0$. Then the operator $X' = X/\lambda$, with $\lambda = \Tr \{ X\} \neq 0$, has trace one. Thus 
\begin{equation}
\Lambda[\lambda X'] = {\lambda} \Lambda[X'] = \lambda \Phi[X'] = \lambda \Psi[X'] = \Psi[X] \; ,
\end{equation} 
implying \eqref{uniqueext} for all $X$ with $\Tr\{X\}\neq 0$.

Let now $\Tr\{X\} = 0$. There exist $X_1, \; X_2 \in \mathcal{B}(\hilbert_S)$ with $\Tr \{X_{1,2}\} \neq 0$ such that $X = X_1 + X_2$. Thus
\begin{equation}
\Lambda[X] = \Lambda[X_1] + \Lambda[X_2] = \Psi[X_1] + \Psi[X_2] = \Psi[X] \; ,
\end{equation} 
which concludes the proof. 
\end{proof}

\subsection{Generic affine subspaces}
We now turn to the more general case of an affine map acting on affine subspaces and study its possible linear extension. We start with a technical Lemma of intuitive content and continue with the main statement.
\begin{mylem}\label{lemma:affelements}
Let $\mathcal{A}$ be an affine subspace of a vector space $V$ (on a field $K$), and let $d$ be the dimension of $\mathcal{A}$. Then there exists a set $\{ \bar{a}_i\} $ of $d+1$ elements of $\mathcal{A}$ such that any element of $\mathcal{A}$ can be written as an affine combination of them with unique coefficients.  
\end{mylem}
\begin{proof}
Let $\alpha\in \mathcal{A}$. Any element $a \in \mathcal{A}$ can be written as $a =  \sum_{i=1}^{d} c_i v_i + \alpha $, with $v_i \in V$ orthogonal and unique coefficients $c_i \in K$. Let now $\bar{a}_i \equiv \alpha + v_i$ for $i=1...d$, and $\bar{a}_{d+1} \equiv \alpha$. Then, $\forall \;a \in \mathcal{A}$
\begin{eqnarray}
a = \sum_{i=1}^{d} c_i \bar{a}_i + \underbrace{(1 - \sum_{i=1}^{d} c_i )}_{=: \; c_{d+1}} \alpha = \sum_{i=1}^{d+1} c_i \bar{a}_i \; ,
\end{eqnarray}
with $\sum_{i=1}^{d+1} c_i =1$. From the uniqueness of the coefficients $\{c_i\}_{i=1}^d$ it follows that all $d+1$ coefficients for each element are uniquely determined by the set $\{\bar{a}_i\}$.
\end{proof}

\begin{mythe}\label{theo:affine2}
An affine map $f: \mathcal{A} \longrightarrow \mathcal{A}'$ between two affine subspaces $\mathcal{A}$ and $\mathcal{A}'$ of vector spaces $V$ and $V'$ can be uniquely extended to a linear map
\begin{equation}
F : \mathrm{Span}(\mathcal{A}) \longrightarrow \mathrm{Span}(\mathcal{A}') \; ,
\end{equation}
where $\mathrm{Span}(\mathcal{A})$ denotes the set of all linear combinations of elements of $\mathcal{A}$.
\end{mythe}
\begin{proof}
Let $v\in \mathrm{Span}(\mathcal{A})$. It can therefore be written as a linear combination of elements in $\mathcal{A}$, $v = \sum_i \beta_i a_i$, with $\beta_i \in K$ and $a_i \in \mathcal{A}$. The map $F(v) := \sum_i \beta_i f(a_i)$ is linear and extends $f$. Suppose there is another set $\{\beta'_i\}$ and $\{ a'_i\}$ such that $v = \sum_i \beta'_i a'_i$. The map $F'(v) := \sum_i \beta'_i f(a'_i)$ is also linear and an extension of $f$. We show that the two maps $F$ and $F'$ always coincide, therefore the linear extension of $f$ is unique.

Per Lemma \ref{lemma:affelements}, 
\begin{eqnarray}
a_i &=&  \sum_{j=1}^{d+1} \lambda_{ij} \bar{a}_j \; ,  \;\;  \sum_{j=1}^{d+1} \lambda_{ij} =1 \;  \forall \; i \\
a'_i &=&  \sum_{j=1}^{d+1} \lambda'_{ij} \bar{a}_j \; ,  \;\;  \sum_{j=1}^{d+1} \lambda'_{ij} =1  \; \forall \; i  \; .
\end{eqnarray}
From the definition of $v$ it follows that $ \sum_{i} \beta_i\lambda_{ij} = \sum_{i} \beta'_i\lambda'_{ij} \; \forall j$.  Now, since $f$ is affine and thus preserves affine combinations, we have
\begin{eqnarray}
F(v)&=&  \sum_i \beta_i f\left(\sum_{j=1}^{d+1} \lambda_{ij} \bar{a}_j \right) = \sum_{ij} \beta_i \lambda_{ij} f(\bar{a}_j) \\
F'(v)&=&   \sum_{ij} \beta'_i \lambda'_{ij} f(\bar{a}_j) \equiv F(v)  \; .
\end{eqnarray}
\end{proof}

To conclude, let us see a case for which the Span of the affine subspace considered is guaranteed to be the complete vector space, as in our case of interest.
\begin{mylem}\label{lemma:span}
Let $\mathcal{A}$ be the affine subspace of vector space $V$ which is generated by the requirement of a fixed (nonzero) outcome to a linear functional $g : V \longrightarrow K$, i.e.
\begin{equation}
\mathcal{A} = \{ v \in V \;| \;g(v) = c\} \; , \; c\neq 0 \;.
\end{equation}
Then, $\mathrm{Span}(\mathcal{A}) = V$. 
\end{mylem}
\begin{proof}
The fact that $\mathrm{Span}(\mathcal{A}) \subset V$ is trivial since $\mathcal{A}\subset V$. For the inverse, let $\{v_i\}$ be any basis of $V$, such that any element can be written as $v= \sum_i \beta_i v_i$.  In general, $g(v_i)= c_i$. But defining $v'_i = c v_i/c_i \in \mathcal{A}$ and $\beta'_i = c_i \beta_i /c$ gives any $v$ as a linear combination of elements in the affine space, therefore $V \subset \mathrm{Span}(\mathcal{A})$.
\end{proof}
Applying Theorem~\ref{theo:affine2} and Lemma~\ref{lemma:span} for $\mathcal{A}=\mathcal{A}^{1}(\hilbert_S)$ one obtains Theorem~\ref{theo:affine1} as a corollary. Furthermore, when Lemma~\ref{lemma:span} applies one can always uniquely extend an affine map 
\begin{equation}
\begin{aligned}
f : \mathcal{A} &\longrightarrow \mathcal{A}'  \\
 a &\longmapsto \tilde{f}(a) + \alpha
\end{aligned}
\end{equation}
with $ \tilde{f}$ linear and $\alpha \in \mathcal{A}'$, from the affine subspace to the full vector space simply by the construction
\begin{equation}
\begin{aligned}
F : V &\longrightarrow \mathrm{Span}(\mathcal{A}')  \\
 v &\longmapsto \tilde{f}(v) +  {g(v)\over c} \alpha \; ,
\end{aligned}
\end{equation}
as was done in Sec.~\ref{sec:linear} using the functional $\Tr\{\cdot\}$ and $c=1$.

\section{Comparison with DSL formalism}\label{app:DSL}
The argument for the existence of the linear dynamical map described in Sec.~\ref{sec:linear} can also be translated into the language of the formalism for linear dynamical maps developed by Dominy, Shabani and Lidar \cite{Dominy2015,Dominy2016}. We briefly summarize the comparison here, along with the point where it slightly deviates. There is a special set of ingredients needed in the DSL approach; we list and examine them for our case in the following.

\emph{1. Convex set of admissible initial states:} as illustrated in Sec.~\ref{subsec:qubits}, different assumptions on the initial total system state lead to different dynamical maps. In the DLS approach one must specify the assumed structure of the total initial state by selecting a (convex) set of possible initial states. In our case, this would read
\begin{equation}
S_{SE} = \{ X\otimes\rho_E +\chi \;  |\;  X \in \mathcal{P}_E^{\chi}(\hilbert_S)\} \;.
\end{equation}
In our formalism, however, we extend this set to the set of operators
\begin{equation}
S^{\chi}_E = \{ X\otimes\rho_E +\chi \;  |\;  X \in \mathcal{A}^1(\hilbert_S)\} \;.
\end{equation}
Notice that the above, while convex as required from the DSL framework,  is not a set of states: $S^{\chi}_E \not\subset \mathcal{S}(\hilbert_{SE})$. 

\emph{2. Linear subspace generated by $S^{\chi}_E$:} the linear space $\mathcal{V}^{\chi}_E = \mathrm{Span}(S^{\chi}_E) \subset \mathcal{B}(\hilbert_{SE})$ is given by 
\begin{equation}
\mathcal{V}^{\chi}_E = \{ X \otimes \rho_E + \chi \Tr\{X\} \;| \;X \in \mathcal{B}(\hilbert_S) \}\; .
\end{equation}
The maps of partial trace over the environment and of unitary evolution upon any element of $S^{\chi}_E$ may then be extended by linearity to any element of $\mathcal{V}^{\chi}_E$.

\emph{3. $U$-consistency:} the set $S^{\chi}_E $ (and consequently also $\mathcal{V}^{\chi}_E$) is $U$-consistent, since whenever 
\begin{equation}
\Tr_E \{ X_1 \otimes \rho_E + \chi \} = \Tr_E \{ X_2 \otimes \rho_E + \chi \} \; , \; \mathrm{i.e.} \; X_1=X_2 \; ,
\end{equation}
it trivially follows that 
\begin{equation}
\Tr_E \{U( X_1 \otimes \rho_E + \chi )U^{\dag}\} =\Tr_E \{U( X_2 \otimes \rho_E + \chi )U^{\dag}\}
\end{equation}
for all unitary operators $U \in \mathrm{U}(\hilbert_{SE})$.
With this, one can extend any dynamical map (i.e.  dynamical maps associated to a unitary operator $U$) acting on $\Tr_E S^{\chi}_E = \mathcal{A}^1(\hilbert_S)$ to a linear map on $\Tr_E \mathcal{V}^{\chi}_E = \mathcal{B}(\hilbert_S)$.

\section{Inhomogeneity from projection operator technique}\label{app:TCLinhomog}
We argue here that the inhomogeneity $\mathcal{J}^{\chi}_t$ in \eqref{corrgenerator} is the same operator that is found through projection operator technique for the derivation of the time convolutionless (TCL) master equation. The argument is given by the fact that the assumptions are the same in both approaches and through the dependency on $\rho_E$ and $\chi$ for different terms of the master equations. 

Notice that in \eqref{corrgenerator}, the uncorrelated generator $\mathcal{L}_t$ depends only on the initial environmental state, and not on $\chi$.  Instead, the inhomogeneity part that was ``linearized'',  $\mathcal{J}^{\chi}_t $, is dependent on both $\rho_E$ and $\chi$. Notice also that this term is linear in $\chi$, so the dependence on $\chi$ cannot be taken away from any part of it. 

Let us now do the same study on the terms of the TCL master equation one derives from \cite{Shibata1977,Chaturvedi1979}. We set the time-independent projection as 
\begin{equation}
P \rho_{SE}(t) = \rho_S(t) \otimes \rho_E \; ,
\end{equation}
i.e. with the initial state of the environment as reference state. It follows that the projection onto the irrelevant part $Q = \mathbb{I} - P$ gives $Q\rho_{SE}(0) = \chi$. Nonetheless, both $P$ and $Q$, as maps, depend only on $\rho_E$.  Without additional assumptions, one gets the equation for the evolution of the relevant part
\begin{equation}\label{relevantme}
\partial_t P \rho_{SE} (t) = \mathcal{K}_t P \rho_{SE} (t) + \mathcal{I}_t \chi \; ,
\end{equation}
with the two linear superoperators
\begin{eqnarray}
\mathcal{K}_t &=& P \mathcal{L}^{SE}_t [ 1 - \Sigma_t]^{-1} P \; , \\
\mathcal{I}_t &=& P \mathcal{L}^{SE}_t [ 1 - \Sigma_t]^{-1} \mathcal{G}_{t,0}Q \; .
\end{eqnarray}
Here appears the generator $ \mathcal{L}^{SE}_t $ for the total system, which depends neither on $\rho_E$ nor on $\chi$. The other operators, instead, are
\begin{eqnarray}
\mathcal{G}_{t,s} &=&  \mathrm{T}_{\leftarrow} \mathrm{exp} \left[ \int_s^t d s' Q \mathcal{L}^{SE}_{s'} \right] \; ,\\
\Sigma_t &=& \int_0^t d s \mathcal{G}_{t,s} Q \mathcal{L}^{SE}_s P G_{t,s} \; \\
G_{t,s} &=&  \mathrm{T}_{\rightarrow} \mathrm{exp} \left[ - \int_s^t d s'  \mathcal{L}^{SE}_{s'} \right] \; 
\end{eqnarray}
where $\mathrm{T}_{\leftarrow}$ and $\mathrm{T}_{\rightarrow}$ indicate respectively the chronological and anti-chronological time ordering operator. The first two operators depend on $\rho_E$ only (through $P$ and $Q$), and the third on neither $\rho_E$ nor $\chi$. This implies that both the maps $\mathcal{K}_t$ and $\mathcal{I}_t$ depend on $\rho_E$ and not on $\chi$.  By applying the left-most projection $P$ in \eqref{relevantme} and tracing out the environment, one gets the exact TCL master equation for the reduced system. From the arguments above, it follows that this master equation is split into a linear superoperator (depending on $\rho_E$ only) acting on $\rho_S(t)$ and an inhomogeneity which depends also on $\rho_E$ as well as -- linearly -- on $\chi$. 

With these consideration we conclude that our term $\mathcal{J}^{\chi}_t$ is indeed the same inhomogeneity that one would obtain from projection operator technique by projecting onto a factorizing state with $\rho_E$ as the time-independent reference state. With this work we show that if one is indeed in possession of the TCL master equation after projecting onto $\rho_S(t) \otimes \rho_E$, the unique linear extension of this inhomogeneous master equation is given by simply multiplying the inhomogeneity by $\Tr\{ \rho_S(t)\}$. Similarly, by finding the eigenvalues and eigenvectors of the inhomogeneity, one can right away put the second term in dissipator form as described in Sec.~\ref{sec:me}.

\end{appendix}

\bibliography{biblio.bib}

\begin{thebibliography}{49}%
\makeatletter
\providecommand \@ifxundefined [1]{%
 \@ifx{#1\undefined}
}%
\providecommand \@ifnum [1]{%
 \ifnum #1\expandafter \@firstoftwo
 \else \expandafter \@secondoftwo
 \fi
}%
\providecommand \@ifx [1]{%
 \ifx #1\expandafter \@firstoftwo
 \else \expandafter \@secondoftwo
 \fi
}%
\providecommand \natexlab [1]{#1}%
\providecommand \enquote  [1]{``#1''}%
\providecommand \bibnamefont  [1]{#1}%
\providecommand \bibfnamefont [1]{#1}%
\providecommand \citenamefont [1]{#1}%
\providecommand \href@noop [0]{\@secondoftwo}%
\providecommand \href [0]{\begingroup \@sanitize@url \@href}%
\providecommand \@href[1]{\@@startlink{#1}\@@href}%
\providecommand \@@href[1]{\endgroup#1\@@endlink}%
\providecommand \@sanitize@url [0]{\catcode `\\12\catcode `\$12\catcode
  `\&12\catcode `\#12\catcode `\^12\catcode `\_12\catcode `\%12\relax}%
\providecommand \@@startlink[1]{}%
\providecommand \@@endlink[0]{}%
\providecommand \url  [0]{\begingroup\@sanitize@url \@url }%
\providecommand \@url [1]{\endgroup\@href {#1}{\urlprefix }}%
\providecommand \urlprefix  [0]{URL }%
\providecommand \Eprint [0]{\href }%
\providecommand \doibase [0]{http://dx.doi.org/}%
\providecommand \selectlanguage [0]{\@gobble}%
\providecommand \bibinfo  [0]{\@secondoftwo}%
\providecommand \bibfield  [0]{\@secondoftwo}%
\providecommand \translation [1]{[#1]}%
\providecommand \BibitemOpen [0]{}%
\providecommand \bibitemStop [0]{}%
\providecommand \bibitemNoStop [0]{.\EOS\space}%
\providecommand \EOS [0]{\spacefactor3000\relax}%
\providecommand \BibitemShut  [1]{\csname bibitem#1\endcsname}%
\let\auto@bib@innerbib\@empty
\bibitem [{\citenamefont {Breuer}\ and\ \citenamefont
  {Petruccione}(2007)}]{Breuer2007}%
  \BibitemOpen
  \bibfield  {author} {\bibinfo {author} {\bibfnamefont {H.-P.}\ \bibnamefont
  {Breuer}}\ and\ \bibinfo {author} {\bibfnamefont {F.}~\bibnamefont
  {Petruccione}},\ }\href {\doibase 10.1093/acprof:oso/9780199213900.001.0001}
  {\emph {\bibinfo {title} {The Theory of Open Quantum Systems}}}\ (\bibinfo
  {publisher} {Oxford University Press},\ \bibinfo {year} {2007})\BibitemShut
  {NoStop}%
\bibitem [{\citenamefont {Nakajima}(1958)}]{Nakajima1958}%
  \BibitemOpen
  \bibfield  {author} {\bibinfo {author} {\bibfnamefont {S.}~\bibnamefont
  {Nakajima}},\ }\href {\doibase 10.1143/ptp.20.948} {\bibfield  {journal}
  {\bibinfo  {journal} {Prog. Theor. Phys.}\ }\textbf {\bibinfo {volume}
  {20}},\ \bibinfo {pages} {948} (\bibinfo {year} {1958})}\BibitemShut
  {NoStop}%
\bibitem [{\citenamefont {Zwanzig}(1960)}]{Zwanzig1960}%
  \BibitemOpen
  \bibfield  {author} {\bibinfo {author} {\bibfnamefont {R.}~\bibnamefont
  {Zwanzig}},\ }\href {\doibase 10.1063/1.1731409} {\bibfield  {journal}
  {\bibinfo  {journal} {J. Chem. Phys.}\ }\textbf {\bibinfo {volume} {33}},\
  \bibinfo {pages} {1338} (\bibinfo {year} {1960})}\BibitemShut {NoStop}%
\bibitem [{\citenamefont {Shibata}\ \emph {et~al.}(1977)\citenamefont
  {Shibata}, \citenamefont {Takahashi},\ and\ \citenamefont
  {Hashitsume}}]{Shibata1977}%
  \BibitemOpen
  \bibfield  {author} {\bibinfo {author} {\bibfnamefont {F.}~\bibnamefont
  {Shibata}}, \bibinfo {author} {\bibfnamefont {Y.}~\bibnamefont {Takahashi}},
  \ and\ \bibinfo {author} {\bibfnamefont {N.}~\bibnamefont {Hashitsume}},\
  }\href {\doibase 10.1007/bf01040100} {\bibfield  {journal} {\bibinfo
  {journal} {J. Stat. Phys.}\ }\textbf {\bibinfo {volume} {17}},\ \bibinfo
  {pages} {171} (\bibinfo {year} {1977})}\BibitemShut {NoStop}%
\bibitem [{\citenamefont {Chaturvedi}\ and\ \citenamefont
  {Shibata}(1979)}]{Chaturvedi1979}%
  \BibitemOpen
  \bibfield  {author} {\bibinfo {author} {\bibfnamefont {S.}~\bibnamefont
  {Chaturvedi}}\ and\ \bibinfo {author} {\bibfnamefont {F.}~\bibnamefont
  {Shibata}},\ }\href {\doibase 10.1007/bf01319852} {\bibfield  {journal}
  {\bibinfo  {journal} {Zeitschrift f\"ur Physik B Condensed Matter}\ }\textbf
  {\bibinfo {volume} {35}},\ \bibinfo {pages} {297} (\bibinfo {year}
  {1979})}\BibitemShut {NoStop}%
\bibitem [{\citenamefont {Munro}\ and\ \citenamefont
  {Gardiner}(1996)}]{Munro1996}%
  \BibitemOpen
  \bibfield  {author} {\bibinfo {author} {\bibfnamefont {W.~J.}\ \bibnamefont
  {Munro}}\ and\ \bibinfo {author} {\bibfnamefont {C.~W.}\ \bibnamefont
  {Gardiner}},\ }\href {\doibase 10.1103/physreva.53.2633} {\bibfield
  {journal} {\bibinfo  {journal} {Phys. Rev. A}\ }\textbf {\bibinfo {volume}
  {53}},\ \bibinfo {pages} {2633} (\bibinfo {year} {1996})}\BibitemShut
  {NoStop}%
\bibitem [{\citenamefont {van Kampen}(2004)}]{Kampen2004}%
  \BibitemOpen
  \bibfield  {author} {\bibinfo {author} {\bibfnamefont {N.~G.}\ \bibnamefont
  {van Kampen}},\ }\href {\doibase 10.1023/b:joss.0000022383.06086.6c}
  {\bibfield  {journal} {\bibinfo  {journal} {J. Stat. Phys.}\ }\textbf
  {\bibinfo {volume} {115}},\ \bibinfo {pages} {1057} (\bibinfo {year}
  {2004})}\BibitemShut {NoStop}%
\bibitem [{\citenamefont {van Kampen}(2005)}]{Kampen2005}%
  \BibitemOpen
  \bibfield  {author} {\bibinfo {author} {\bibfnamefont {N.~G.}\ \bibnamefont
  {van Kampen}},\ }\href {\doibase 10.1021/jp0581633} {\bibfield  {journal}
  {\bibinfo  {journal} {J. Phys. Chem. B}\ }\textbf {\bibinfo {volume} {109}},\
  \bibinfo {pages} {21293} (\bibinfo {year} {2005})}\BibitemShut {NoStop}%
\bibitem [{\citenamefont {Pechukas}(1994)}]{Pechukas1994}%
  \BibitemOpen
  \bibfield  {author} {\bibinfo {author} {\bibfnamefont {P.}~\bibnamefont
  {Pechukas}},\ }\href {\doibase 10.1103/physrevlett.73.1060} {\bibfield
  {journal} {\bibinfo  {journal} {Phys. Rev. Lett.}\ }\textbf {\bibinfo
  {volume} {73}},\ \bibinfo {pages} {1060} (\bibinfo {year}
  {1994})}\BibitemShut {NoStop}%
\bibitem [{\citenamefont {Alicki}(1995)}]{Alicki1995}%
  \BibitemOpen
  \bibfield  {author} {\bibinfo {author} {\bibfnamefont {R.}~\bibnamefont
  {Alicki}},\ }\href {\doibase 10.1103/physrevlett.75.3020} {\bibfield
  {journal} {\bibinfo  {journal} {Phys. Rev. Lett.}\ }\textbf {\bibinfo
  {volume} {75}},\ \bibinfo {pages} {3020} (\bibinfo {year}
  {1995})}\BibitemShut {NoStop}%
\bibitem [{\citenamefont {Pechukas}(1995)}]{Pechukas1995}%
  \BibitemOpen
  \bibfield  {author} {\bibinfo {author} {\bibfnamefont {P.}~\bibnamefont
  {Pechukas}},\ }\href {\doibase 10.1103/PhysRevLett.75.3021} {\bibfield
  {journal} {\bibinfo  {journal} {Phys. Rev. Lett.}\ }\textbf {\bibinfo
  {volume} {75}},\ \bibinfo {pages} {3021} (\bibinfo {year}
  {1995})}\BibitemShut {NoStop}%
\bibitem [{\citenamefont {Laine}\ \emph {et~al.}(2010)\citenamefont {Laine},
  \citenamefont {Piilo},\ and\ \citenamefont {Breuer}}]{Laine2010b}%
  \BibitemOpen
  \bibfield  {author} {\bibinfo {author} {\bibfnamefont {E.-M.}\ \bibnamefont
  {Laine}}, \bibinfo {author} {\bibfnamefont {J.}~\bibnamefont {Piilo}}, \ and\
  \bibinfo {author} {\bibfnamefont {H.-P.}\ \bibnamefont {Breuer}},\
  }\href@noop {} {\bibfield  {journal} {\bibinfo  {journal} {EPL}\ }\textbf
  {\bibinfo {volume} {92}},\ \bibinfo {pages} {60010} (\bibinfo {year}
  {2010})}\BibitemShut {NoStop}%
\bibitem [{\citenamefont {Gessner}\ and\ \citenamefont
  {Breuer}(2011)}]{Gessner2011a}%
  \BibitemOpen
  \bibfield  {author} {\bibinfo {author} {\bibfnamefont {M.}~\bibnamefont
  {Gessner}}\ and\ \bibinfo {author} {\bibfnamefont {H.-P.}\ \bibnamefont
  {Breuer}},\ }\href {\doibase 10.1103/PhysRevLett.107.180402} {\bibfield
  {journal} {\bibinfo  {journal} {Phys. Rev. Lett.}\ }\textbf {\bibinfo
  {volume} {107}},\ \bibinfo {pages} {180402} (\bibinfo {year}
  {2011})}\BibitemShut {NoStop}%
\bibitem [{\citenamefont {Gessner}\ \emph {et~al.}(2014)\citenamefont
  {Gessner}, \citenamefont {Ramm}, \citenamefont {Pruttivarasin}, \citenamefont
  {Buchleitner}, \citenamefont {Breuer},\ and\ \citenamefont
  {H{\"a}ffner}}]{Gessner2014a}%
  \BibitemOpen
  \bibfield  {author} {\bibinfo {author} {\bibfnamefont {M.}~\bibnamefont
  {Gessner}}, \bibinfo {author} {\bibfnamefont {M.}~\bibnamefont {Ramm}},
  \bibinfo {author} {\bibfnamefont {T.}~\bibnamefont {Pruttivarasin}}, \bibinfo
  {author} {\bibfnamefont {A.}~\bibnamefont {Buchleitner}}, \bibinfo {author}
  {\bibfnamefont {H.-P.}\ \bibnamefont {Breuer}}, \ and\ \bibinfo {author}
  {\bibfnamefont {H.}~\bibnamefont {H{\"a}ffner}},\ }\href@noop {} {\bibfield
  {journal} {\bibinfo  {journal} {Nat. Phys.}\ }\textbf {\bibinfo {volume}
  {10}},\ \bibinfo {pages} {105} (\bibinfo {year} {2014})}\BibitemShut
  {NoStop}%
\bibitem [{\citenamefont {Cialdi}\ \emph {et~al.}(2014)\citenamefont {Cialdi},
  \citenamefont {Smirne}, \citenamefont {Paris}, \citenamefont {Olivares},\
  and\ \citenamefont {Vacchini}}]{Cialdi2014}%
  \BibitemOpen
  \bibfield  {author} {\bibinfo {author} {\bibfnamefont {S.}~\bibnamefont
  {Cialdi}}, \bibinfo {author} {\bibfnamefont {A.}~\bibnamefont {Smirne}},
  \bibinfo {author} {\bibfnamefont {M.~G.~A.}\ \bibnamefont {Paris}}, \bibinfo
  {author} {\bibfnamefont {S.}~\bibnamefont {Olivares}}, \ and\ \bibinfo
  {author} {\bibfnamefont {B.}~\bibnamefont {Vacchini}},\ }\href {\doibase
  10.1103/PhysRevA.90.050301} {\bibfield  {journal} {\bibinfo  {journal} {Phys.
  Rev. A}\ }\textbf {\bibinfo {volume} {90}},\ \bibinfo {pages} {050301}
  (\bibinfo {year} {2014})}\BibitemShut {NoStop}%
\bibitem [{\citenamefont {Gessner}\ and\ \citenamefont
  {Breuer}(2019)}]{Gessner2019}%
  \BibitemOpen
  \bibfield  {author} {\bibinfo {author} {\bibfnamefont {M.}~\bibnamefont
  {Gessner}}\ and\ \bibinfo {author} {\bibfnamefont {H.-P.}\ \bibnamefont
  {Breuer}},\ }\enquote {\bibinfo {title} {{Revealing Correlations Between a
  System and an Inaccessible Environment}},}\ in\ \href@noop {} {\emph
  {\bibinfo {booktitle} {Advances in Open Systems and Fundamental Tests of
  Quantum Mechanics}}},\ \bibinfo {series and number} {Springer Proceedings in
  Physics 237},\ \bibinfo {editor} {edited by\ \bibinfo {editor} {\bibfnamefont
  {B.}~\bibnamefont {Vacchini}}, \bibinfo {editor} {\bibfnamefont {H.-P.}\
  \bibnamefont {Breuer}}, \ and\ \bibinfo {editor} {\bibfnamefont
  {A.}~\bibnamefont {Bassi}}}\ (\bibinfo  {publisher} {Springer},\ \bibinfo
  {address} {Cham},\ \bibinfo {year} {2019})\ pp.\ \bibinfo {pages} {59 --
  71}\BibitemShut {NoStop}%
\bibitem [{\citenamefont {Shabani}\ and\ \citenamefont
  {Lidar}(2009)}]{Shabani2009}%
  \BibitemOpen
  \bibfield  {author} {\bibinfo {author} {\bibfnamefont {A.}~\bibnamefont
  {Shabani}}\ and\ \bibinfo {author} {\bibfnamefont {D.~A.}\ \bibnamefont
  {Lidar}},\ }\href {\doibase 10.1103/physrevlett.102.100402} {\bibfield
  {journal} {\bibinfo  {journal} {Phys. Rev. Lett.}\ }\textbf {\bibinfo
  {volume} {102}} (\bibinfo {year} {2009}),\
  10.1103/physrevlett.102.100402}\BibitemShut {NoStop}%
\bibitem [{\citenamefont {Brodutch}\ \emph {et~al.}(2013)\citenamefont
  {Brodutch}, \citenamefont {Datta}, \citenamefont {Modi}, \citenamefont
  {Rivas},\ and\ \citenamefont {Rodr{\'{\i}}guez-Rosario}}]{Brodutch2013}%
  \BibitemOpen
  \bibfield  {author} {\bibinfo {author} {\bibfnamefont {A.}~\bibnamefont
  {Brodutch}}, \bibinfo {author} {\bibfnamefont {A.}~\bibnamefont {Datta}},
  \bibinfo {author} {\bibfnamefont {K.}~\bibnamefont {Modi}}, \bibinfo {author}
  {\bibfnamefont {{\'{A}}.}~\bibnamefont {Rivas}}, \ and\ \bibinfo {author}
  {\bibfnamefont {C.~A.}\ \bibnamefont {Rodr{\'{\i}}guez-Rosario}},\ }\href
  {\doibase 10.1103/physreva.87.042301} {\bibfield  {journal} {\bibinfo
  {journal} {Phys. Rev. A}\ }\textbf {\bibinfo {volume} {87}} (\bibinfo {year}
  {2013}),\ 10.1103/physreva.87.042301}\BibitemShut {NoStop}%
\bibitem [{\citenamefont {Liu}\ and\ \citenamefont {Tong}(2014)}]{Liu2014}%
  \BibitemOpen
  \bibfield  {author} {\bibinfo {author} {\bibfnamefont {L.}~\bibnamefont
  {Liu}}\ and\ \bibinfo {author} {\bibfnamefont {D.~M.}\ \bibnamefont {Tong}},\
  }\href {\doibase 10.1103/physreva.90.012305} {\bibfield  {journal} {\bibinfo
  {journal} {Phys. Rev. A}\ }\textbf {\bibinfo {volume} {90}} (\bibinfo {year}
  {2014}),\ 10.1103/physreva.90.012305}\BibitemShut {NoStop}%
\bibitem [{\citenamefont {Buscemi}(2014)}]{Buscemi2014}%
  \BibitemOpen
  \bibfield  {author} {\bibinfo {author} {\bibfnamefont {F.}~\bibnamefont
  {Buscemi}},\ }\href {\doibase 10.1103/physrevlett.113.140502} {\bibfield
  {journal} {\bibinfo  {journal} {Phys. Rev. Lett.}\ }\textbf {\bibinfo
  {volume} {113}} (\bibinfo {year} {2014}),\
  10.1103/physrevlett.113.140502}\BibitemShut {NoStop}%
\bibitem [{\citenamefont {Shabani}\ and\ \citenamefont
  {Lidar}(2016)}]{Shabani2016}%
  \BibitemOpen
  \bibfield  {author} {\bibinfo {author} {\bibfnamefont {A.}~\bibnamefont
  {Shabani}}\ and\ \bibinfo {author} {\bibfnamefont {D.~A.}\ \bibnamefont
  {Lidar}},\ }\href {\doibase 10.1103/physrevlett.116.049901} {\bibfield
  {journal} {\bibinfo  {journal} {Phys. Rev. Lett.}\ }\textbf {\bibinfo
  {volume} {116}} (\bibinfo {year} {2016}),\
  10.1103/physrevlett.116.049901}\BibitemShut {NoStop}%
\bibitem [{\citenamefont {Vacchini}\ and\ \citenamefont
  {Amato}(2016)}]{Vacchini2016}%
  \BibitemOpen
  \bibfield  {author} {\bibinfo {author} {\bibfnamefont {B.}~\bibnamefont
  {Vacchini}}\ and\ \bibinfo {author} {\bibfnamefont {G.}~\bibnamefont
  {Amato}},\ }\href {\doibase 10.1038/srep37328} {\bibfield  {journal}
  {\bibinfo  {journal} {Sci. Rep.}\ }\textbf {\bibinfo {volume} {6}} (\bibinfo
  {year} {2016}),\ 10.1038/srep37328}\BibitemShut {NoStop}%
\bibitem [{\citenamefont {Schmid}\ \emph {et~al.}(2019)\citenamefont {Schmid},
  \citenamefont {Ried},\ and\ \citenamefont {Spekkens}}]{Schmid2019}%
  \BibitemOpen
  \bibfield  {author} {\bibinfo {author} {\bibfnamefont {D.}~\bibnamefont
  {Schmid}}, \bibinfo {author} {\bibfnamefont {K.}~\bibnamefont {Ried}}, \ and\
  \bibinfo {author} {\bibfnamefont {R.~W.}\ \bibnamefont {Spekkens}},\ }\href
  {\doibase 10.1103/physreva.100.022112} {\bibfield  {journal} {\bibinfo
  {journal} {Phys. Rev. A}\ }\textbf {\bibinfo {volume} {100}} (\bibinfo {year}
  {2019}),\ 10.1103/physreva.100.022112}\BibitemShut {NoStop}%
\bibitem [{\citenamefont {Modi}(2012)}]{Modi2012}%
  \BibitemOpen
  \bibfield  {author} {\bibinfo {author} {\bibfnamefont {K.}~\bibnamefont
  {Modi}},\ }\href {\doibase 10.1038/srep00581} {\bibfield  {journal} {\bibinfo
   {journal} {Scientific Reports}\ }\textbf {\bibinfo {volume} {2}},\ \bibinfo
  {pages} {581} (\bibinfo {year} {2012})}\BibitemShut {NoStop}%
\bibitem [{\citenamefont {Ringbauer}\ \emph {et~al.}(2015)\citenamefont
  {Ringbauer}, \citenamefont {Wood}, \citenamefont {Modi}, \citenamefont
  {Gilchrist}, \citenamefont {White},\ and\ \citenamefont
  {Fedrizzi}}]{Ringbauer2015}%
  \BibitemOpen
  \bibfield  {author} {\bibinfo {author} {\bibfnamefont {M.}~\bibnamefont
  {Ringbauer}}, \bibinfo {author} {\bibfnamefont {C.~J.}\ \bibnamefont {Wood}},
  \bibinfo {author} {\bibfnamefont {K.}~\bibnamefont {Modi}}, \bibinfo {author}
  {\bibfnamefont {A.}~\bibnamefont {Gilchrist}}, \bibinfo {author}
  {\bibfnamefont {A.~G.}\ \bibnamefont {White}}, \ and\ \bibinfo {author}
  {\bibfnamefont {A.}~\bibnamefont {Fedrizzi}},\ }\href {\doibase
  10.1103/PhysRevLett.114.090402} {\bibfield  {journal} {\bibinfo  {journal}
  {Phys. Rev. Lett.}\ }\textbf {\bibinfo {volume} {114}},\ \bibinfo {pages}
  {090402} (\bibinfo {year} {2015})}\BibitemShut {NoStop}%
\bibitem [{\citenamefont {Paz-Silva}\ \emph {et~al.}(2019)\citenamefont
  {Paz-Silva}, \citenamefont {Hall},\ and\ \citenamefont
  {Wiseman}}]{PazSilva2019}%
  \BibitemOpen
  \bibfield  {author} {\bibinfo {author} {\bibfnamefont {G.~A.}\ \bibnamefont
  {Paz-Silva}}, \bibinfo {author} {\bibfnamefont {M.~J.~W.}\ \bibnamefont
  {Hall}}, \ and\ \bibinfo {author} {\bibfnamefont {H.~M.}\ \bibnamefont
  {Wiseman}},\ }\href {\doibase 10.1103/physreva.100.042120} {\bibfield
  {journal} {\bibinfo  {journal} {Phys. Rev. A}\ }\textbf {\bibinfo {volume}
  {100}} (\bibinfo {year} {2019}),\ 10.1103/physreva.100.042120}\BibitemShut
  {NoStop}%
\bibitem [{\citenamefont {Alipour}\ \emph {et~al.}(2020)\citenamefont
  {Alipour}, \citenamefont {Rezakhani}, \citenamefont {Babu}, \citenamefont
  {M{\o}lmer}, \citenamefont {M\"ott\"onen},\ and\ \citenamefont
  {Ala-Nissila}}]{Alipour2020}%
  \BibitemOpen
  \bibfield  {author} {\bibinfo {author} {\bibfnamefont {S.}~\bibnamefont
  {Alipour}}, \bibinfo {author} {\bibfnamefont {A.~T.}\ \bibnamefont
  {Rezakhani}}, \bibinfo {author} {\bibfnamefont {A.~P.}\ \bibnamefont {Babu}},
  \bibinfo {author} {\bibfnamefont {K.}~\bibnamefont {M{\o}lmer}}, \bibinfo
  {author} {\bibfnamefont {M.}~\bibnamefont {M\"ott\"onen}}, \ and\ \bibinfo
  {author} {\bibfnamefont {T.}~\bibnamefont {Ala-Nissila}},\ }\href {\doibase
  10.1103/physrevx.10.041024} {\bibfield  {journal} {\bibinfo  {journal} {Phys.
  Rev. X}\ }\textbf {\bibinfo {volume} {10}} (\bibinfo {year} {2020}),\
  10.1103/physrevx.10.041024}\BibitemShut {NoStop}%
\bibitem [{\citenamefont {Trevisan}\ \emph {et~al.}(2021)\citenamefont
  {Trevisan}, \citenamefont {Smirne}, \citenamefont {Megier},\ and\
  \citenamefont {Vacchini}}]{Trevisan2021}%
  \BibitemOpen
  \bibfield  {author} {\bibinfo {author} {\bibfnamefont {A.}~\bibnamefont
  {Trevisan}}, \bibinfo {author} {\bibfnamefont {A.}~\bibnamefont {Smirne}},
  \bibinfo {author} {\bibfnamefont {N.}~\bibnamefont {Megier}}, \ and\ \bibinfo
  {author} {\bibfnamefont {B.}~\bibnamefont {Vacchini}},\ }\href {\doibase
  10.1103/physreva.104.052215} {\bibfield  {journal} {\bibinfo  {journal}
  {Phys. Rev. A}\ }\textbf {\bibinfo {volume} {104}} (\bibinfo {year} {2021}),\
  10.1103/physreva.104.052215}\BibitemShut {NoStop}%
\bibitem [{\citenamefont {{\v{S}}telmachovi{\v{c}}}\ and\ \citenamefont
  {Bu{\v{z}}ek}(2001)}]{Stelmachovic2001}%
  \BibitemOpen
  \bibfield  {author} {\bibinfo {author} {\bibfnamefont {P.}~\bibnamefont
  {{\v{S}}telmachovi{\v{c}}}}\ and\ \bibinfo {author} {\bibfnamefont
  {V.}~\bibnamefont {Bu{\v{z}}ek}},\ }\href {\doibase
  10.1103/physreva.64.062106} {\bibfield  {journal} {\bibinfo  {journal} {Phys.
  Rev. A}\ }\textbf {\bibinfo {volume} {64}} (\bibinfo {year} {2001}),\
  10.1103/physreva.64.062106}\BibitemShut {NoStop}%
\bibitem [{\citenamefont {Zhang}\ \emph {et~al.}(2015)\citenamefont {Zhang},
  \citenamefont {Liu},\ and\ \citenamefont {Tong}}]{Zhang2015}%
  \BibitemOpen
  \bibfield  {author} {\bibinfo {author} {\bibfnamefont {D.-J.}\ \bibnamefont
  {Zhang}}, \bibinfo {author} {\bibfnamefont {C.-L.}\ \bibnamefont {Liu}}, \
  and\ \bibinfo {author} {\bibfnamefont {D.-M.}\ \bibnamefont {Tong}},\ }\href
  {\doibase 10.1088/0256-307x/32/4/040302} {\bibfield  {journal} {\bibinfo
  {journal} {Chin. Phys. Lett.}\ }\textbf {\bibinfo {volume} {32}},\ \bibinfo
  {pages} {040302} (\bibinfo {year} {2015})}\BibitemShut {NoStop}%
\bibitem [{\citenamefont {Jordan}\ \emph {et~al.}(2004)\citenamefont {Jordan},
  \citenamefont {Shaji},\ and\ \citenamefont {Sudarshan}}]{Jordan2004}%
  \BibitemOpen
  \bibfield  {author} {\bibinfo {author} {\bibfnamefont {T.~F.}\ \bibnamefont
  {Jordan}}, \bibinfo {author} {\bibfnamefont {A.}~\bibnamefont {Shaji}}, \
  and\ \bibinfo {author} {\bibfnamefont {E.~C.~G.}\ \bibnamefont {Sudarshan}},\
  }\href {\doibase 10.1103/physreva.70.052110} {\bibfield  {journal} {\bibinfo
  {journal} {Phys. Rev. A}\ }\textbf {\bibinfo {volume} {70}} (\bibinfo {year}
  {2004}),\ 10.1103/physreva.70.052110}\BibitemShut {NoStop}%
\bibitem [{\citenamefont {Jordan}(2005)}]{Jordan2005}%
  \BibitemOpen
  \bibfield  {author} {\bibinfo {author} {\bibfnamefont {T.~F.}\ \bibnamefont
  {Jordan}},\ }\href {\doibase 10.1103/physreva.71.034101} {\bibfield
  {journal} {\bibinfo  {journal} {Phys. Rev. A}\ }\textbf {\bibinfo {volume}
  {71}} (\bibinfo {year} {2005}),\ 10.1103/physreva.71.034101}\BibitemShut
  {NoStop}%
\bibitem [{\citenamefont {Dominy}\ \emph {et~al.}(2015)\citenamefont {Dominy},
  \citenamefont {Shabani},\ and\ \citenamefont {Lidar}}]{Dominy2015}%
  \BibitemOpen
  \bibfield  {author} {\bibinfo {author} {\bibfnamefont {J.~M.}\ \bibnamefont
  {Dominy}}, \bibinfo {author} {\bibfnamefont {A.}~\bibnamefont {Shabani}}, \
  and\ \bibinfo {author} {\bibfnamefont {D.~A.}\ \bibnamefont {Lidar}},\ }\href
  {\doibase 10.1007/s11128-015-1148-0} {\bibfield  {journal} {\bibinfo
  {journal} {Quantum Inf. Process.}\ }\textbf {\bibinfo {volume} {15}},\
  \bibinfo {pages} {465} (\bibinfo {year} {2015})}\BibitemShut {NoStop}%
\bibitem [{\citenamefont {Dominy}\ and\ \citenamefont
  {Lidar}(2016)}]{Dominy2016}%
  \BibitemOpen
  \bibfield  {author} {\bibinfo {author} {\bibfnamefont {J.~M.}\ \bibnamefont
  {Dominy}}\ and\ \bibinfo {author} {\bibfnamefont {D.~A.}\ \bibnamefont
  {Lidar}},\ }\href {\doibase 10.1007/s11128-015-1228-1} {\bibfield  {journal}
  {\bibinfo  {journal} {Quantum Inf. Process.}\ }\textbf {\bibinfo {volume}
  {15}},\ \bibinfo {pages} {1349} (\bibinfo {year} {2016})}\BibitemShut
  {NoStop}%
\bibitem [{\citenamefont {Choi}(1975)}]{Choi1975}%
  \BibitemOpen
  \bibfield  {author} {\bibinfo {author} {\bibfnamefont {M.-D.}\ \bibnamefont
  {Choi}},\ }\href {\doibase https://doi.org/10.1016/0024-3795(75)90075-0}
  {\bibfield  {journal} {\bibinfo  {journal} {Linear Algebra Appl.}\ }\textbf
  {\bibinfo {volume} {10}},\ \bibinfo {pages} {285} (\bibinfo {year}
  {1975})}\BibitemShut {NoStop}%
\bibitem [{\citenamefont {Kraus}(1983)}]{Kraus1983}%
  \BibitemOpen
  \bibfield  {author} {\bibinfo {author} {\bibfnamefont {K.}~\bibnamefont
  {Kraus}},\ }\href {\doibase 10.1007/3-540-12732-1} {\emph {\bibinfo {title}
  {States, Effects, and Operations Fundamental Notions of Quantum Theory}}}\
  (\bibinfo  {publisher} {Springer Berlin Heidelberg},\ \bibinfo {year}
  {1983})\BibitemShut {NoStop}%
\bibitem [{\citenamefont {Breuer}(2012)}]{Breuer2012}%
  \BibitemOpen
  \bibfield  {author} {\bibinfo {author} {\bibfnamefont {H.-P.}\ \bibnamefont
  {Breuer}},\ }\href {\doibase 10.1088/0953-4075/45/15/154001} {\bibfield
  {journal} {\bibinfo  {journal} {J. Phys. B: At. Mol. Opt. Phys.}\ }\textbf
  {\bibinfo {volume} {45}},\ \bibinfo {pages} {154001} (\bibinfo {year}
  {2012})}\BibitemShut {NoStop}%
\bibitem [{\citenamefont {Hall}\ \emph {et~al.}(2014)\citenamefont {Hall},
  \citenamefont {Cresser}, \citenamefont {Li},\ and\ \citenamefont
  {Andersson}}]{Hall2014}%
  \BibitemOpen
  \bibfield  {author} {\bibinfo {author} {\bibfnamefont {M.~J.~W.}\
  \bibnamefont {Hall}}, \bibinfo {author} {\bibfnamefont {J.~D.}\ \bibnamefont
  {Cresser}}, \bibinfo {author} {\bibfnamefont {L.}~\bibnamefont {Li}}, \ and\
  \bibinfo {author} {\bibfnamefont {E.}~\bibnamefont {Andersson}},\ }\href
  {\doibase 10.1103/physreva.89.042120} {\bibfield  {journal} {\bibinfo
  {journal} {Phys. Rev. A}\ }\textbf {\bibinfo {volume} {89}} (\bibinfo {year}
  {2014}),\ 10.1103/physreva.89.042120}\BibitemShut {NoStop}%
\bibitem [{\citenamefont {Breuer}\ \emph {et~al.}(1999)\citenamefont {Breuer},
  \citenamefont {Kappler},\ and\ \citenamefont {Petruccione}}]{Breuer1999b}%
  \BibitemOpen
  \bibfield  {author} {\bibinfo {author} {\bibfnamefont {H.-P.}\ \bibnamefont
  {Breuer}}, \bibinfo {author} {\bibfnamefont {B.}~\bibnamefont {Kappler}}, \
  and\ \bibinfo {author} {\bibfnamefont {F.}~\bibnamefont {Petruccione}},\
  }\href@noop {} {\bibfield  {journal} {\bibinfo  {journal} {Phys. Rev. A}\
  }\textbf {\bibinfo {volume} {59}},\ \bibinfo {pages} {1633} (\bibinfo {year}
  {1999})}\BibitemShut {NoStop}%
\bibitem [{\citenamefont {Breuer}(2004)}]{Breuer2004}%
  \BibitemOpen
  \bibfield  {author} {\bibinfo {author} {\bibfnamefont {H.-P.}\ \bibnamefont
  {Breuer}},\ }\href {\doibase 10.1103/PhysRevA.70.012106} {\bibfield
  {journal} {\bibinfo  {journal} {Phys. Rev. A}\ }\textbf {\bibinfo {volume}
  {70}},\ \bibinfo {pages} {012106} (\bibinfo {year} {2004})}\BibitemShut
  {NoStop}%
\bibitem [{\citenamefont {Piilo}\ \emph {et~al.}(2008)\citenamefont {Piilo},
  \citenamefont {Maniscalco}, \citenamefont {Harkonen},\ and\ \citenamefont
  {Suominen}}]{Piilo2008a}%
  \BibitemOpen
  \bibfield  {author} {\bibinfo {author} {\bibfnamefont {J.}~\bibnamefont
  {Piilo}}, \bibinfo {author} {\bibfnamefont {S.}~\bibnamefont {Maniscalco}},
  \bibinfo {author} {\bibfnamefont {K.}~\bibnamefont {Harkonen}}, \ and\
  \bibinfo {author} {\bibfnamefont {K.-A.}\ \bibnamefont {Suominen}},\
  }\href@noop {} {\bibfield  {journal} {\bibinfo  {journal} {Phys. Rev. Lett.}\
  }\textbf {\bibinfo {volume} {100}},\ \bibinfo {eid} {180402} (\bibinfo {year}
  {2008})}\BibitemShut {NoStop}%
\bibitem [{\citenamefont {Smirne}\ \emph {et~al.}(2020)\citenamefont {Smirne},
  \citenamefont {Caiaffa},\ and\ \citenamefont {Piilo}}]{Smirne2020}%
  \BibitemOpen
  \bibfield  {author} {\bibinfo {author} {\bibfnamefont {A.}~\bibnamefont
  {Smirne}}, \bibinfo {author} {\bibfnamefont {M.}~\bibnamefont {Caiaffa}}, \
  and\ \bibinfo {author} {\bibfnamefont {J.}~\bibnamefont {Piilo}},\ }\href
  {\doibase 10.1103/PhysRevLett.124.190402} {\bibfield  {journal} {\bibinfo
  {journal} {Phys. Rev. Lett.}\ }\textbf {\bibinfo {volume} {124}},\ \bibinfo
  {pages} {190402} (\bibinfo {year} {2020})}\BibitemShut {NoStop}%
\bibitem [{\citenamefont {Colla}\ and\ \citenamefont
  {Breuer}(2022)}]{Colla2022}%
  \BibitemOpen
  \bibfield  {author} {\bibinfo {author} {\bibfnamefont {A.}~\bibnamefont
  {Colla}}\ and\ \bibinfo {author} {\bibfnamefont {H.-P.}\ \bibnamefont
  {Breuer}},\ }\href {\doibase 10.1103/PhysRevA.105.052216} {\bibfield
  {journal} {\bibinfo  {journal} {Phys. Rev. A}\ }\textbf {\bibinfo {volume}
  {105}},\ \bibinfo {pages} {052216} (\bibinfo {year} {2022})}\BibitemShut
  {NoStop}%
\bibitem [{\citenamefont {Sargolzahi}(2020)}]{Sargolzahi2020}%
  \BibitemOpen
  \bibfield  {author} {\bibinfo {author} {\bibfnamefont {I.}~\bibnamefont
  {Sargolzahi}},\ }\href {\doibase 10.1103/physreva.102.022208} {\bibfield
  {journal} {\bibinfo  {journal} {Phys. Rev. A}\ }\textbf {\bibinfo {volume}
  {102}} (\bibinfo {year} {2020}),\ 10.1103/physreva.102.022208}\BibitemShut
  {NoStop}%
\bibitem [{\citenamefont {Rodr{\'{\i}}guez-Rosario}\ \emph
  {et~al.}(2008)\citenamefont {Rodr{\'{\i}}guez-Rosario}, \citenamefont {Modi},
  \citenamefont {meng Kuah}, \citenamefont {Shaji},\ and\ \citenamefont
  {Sudarshan}}]{RodriguezRosario2008}%
  \BibitemOpen
  \bibfield  {author} {\bibinfo {author} {\bibfnamefont {C.~A.}\ \bibnamefont
  {Rodr{\'{\i}}guez-Rosario}}, \bibinfo {author} {\bibfnamefont
  {K.}~\bibnamefont {Modi}}, \bibinfo {author} {\bibfnamefont {A.}~\bibnamefont
  {meng Kuah}}, \bibinfo {author} {\bibfnamefont {A.}~\bibnamefont {Shaji}}, \
  and\ \bibinfo {author} {\bibfnamefont {E.~C.~G.}\ \bibnamefont {Sudarshan}},\
  }\href {\doibase 10.1088/1751-8113/41/20/205301} {\bibfield  {journal}
  {\bibinfo  {journal} {J. Phys. A: Math. Theor.}\ }\textbf {\bibinfo {volume}
  {41}},\ \bibinfo {pages} {205301} (\bibinfo {year} {2008})}\BibitemShut
  {NoStop}%
\bibitem [{\citenamefont {Sorce}\ and\ \citenamefont
  {Hayden}(2022)}]{Sorce2022}%
  \BibitemOpen
  \bibfield  {author} {\bibinfo {author} {\bibfnamefont {J.}~\bibnamefont
  {Sorce}}\ and\ \bibinfo {author} {\bibfnamefont {P.~M.}\ \bibnamefont
  {Hayden}},\ }\href {\doibase 10.1088/1751-8121/ac65c2} {\bibfield  {journal}
  {\bibinfo  {journal} {J. Phys. A: Math. Theor.}\ } (\bibinfo {year} {2022}),\
  10.1088/1751-8121/ac65c2}\BibitemShut {NoStop}%
\bibitem [{\citenamefont {Jaynes}\ and\ \citenamefont
  {Cummings}(1963)}]{Jaynes1963}%
  \BibitemOpen
  \bibfield  {author} {\bibinfo {author} {\bibfnamefont {E.}~\bibnamefont
  {Jaynes}}\ and\ \bibinfo {author} {\bibfnamefont {F.}~\bibnamefont
  {Cummings}},\ }\href {\doibase 10.1109/PROC.1963.1664} {\bibfield  {journal}
  {\bibinfo  {journal} {Proceedings of the IEEE}\ }\textbf {\bibinfo {volume}
  {51}},\ \bibinfo {pages} {89} (\bibinfo {year} {1963})}\BibitemShut {NoStop}%
\bibitem [{\citenamefont {Grynberg}\ \emph {et~al.}(2010)\citenamefont
  {Grynberg}, \citenamefont {Aspect},\ and\ \citenamefont
  {Fabre}}]{grynberg2010}%
  \BibitemOpen
  \bibfield  {author} {\bibinfo {author} {\bibfnamefont {G.}~\bibnamefont
  {Grynberg}}, \bibinfo {author} {\bibfnamefont {A.}~\bibnamefont {Aspect}}, \
  and\ \bibinfo {author} {\bibfnamefont {C.}~\bibnamefont {Fabre}},\
  }\href@noop {} {\emph {\bibinfo {title} {Introduction to quantum optics: from
  the semi-classical approach to quantized light}}}\ (\bibinfo  {publisher}
  {Cambridge university press},\ \bibinfo {year} {2010})\BibitemShut {NoStop}%
\bibitem [{\citenamefont {Smirne}\ and\ \citenamefont
  {Vacchini}(2010)}]{Smirne2010b}%
  \BibitemOpen
  \bibfield  {author} {\bibinfo {author} {\bibfnamefont {A.}~\bibnamefont
  {Smirne}}\ and\ \bibinfo {author} {\bibfnamefont {B.}~\bibnamefont
  {Vacchini}},\ }\href {\doibase 10.1103/PhysRevA.82.022110} {\bibfield
  {journal} {\bibinfo  {journal} {Phys. Rev. A}\ }\textbf {\bibinfo {volume}
  {82}},\ \bibinfo {pages} {022110} (\bibinfo {year} {2010})}\BibitemShut
  {NoStop}%
\end{thebibliography}%

\end{document}